\documentclass[acmsmall, nonacm]{acmart}
\AtBeginDocument{%
  }

\usepackage[T1]{fontenc}
\usepackage{microtype}
\usepackage{graphicx}
\usepackage{booktabs} 

\usepackage{hyperref}

\usepackage{amsmath}
\usepackage{mathtools}

\usepackage{bbm}

\usepackage{mdframed}
\mdfdefinestyle{Highlight}{%
    linecolor=blue,
    outerlinewidth=2pt,
    roundcorner=20pt,
    innertopmargin=\baselineskip,
    innerbottommargin=\baselineskip,
    innerrightmargin=20pt,
    innerleftmargin=20pt,
    backgroundcolor=gray!50!white}

\usepackage[linesnumbered, ruled]{algorithm2e}
\SetKwRepeat{Do}{do}{while}%
\usepackage{algpseudocode}

\usepackage[capitalize,noabbrev]{cleveref}

\usepackage{wrapfig}

\usepackage{listings}

\lstdefinestyle{mypython}{
    language=Python,
    backgroundcolor=\color{white},
    basicstyle=\ttfamily,
    commentstyle=\color{gray},
    keywordstyle=\color{blue},
    stringstyle=\color{red},
    showstringspaces=false,
    numbers=left,
    numberstyle=\tiny\color{gray},
    stepnumber=1,
    frame=single,
    tabsize=4,
    morekeywords={Server, Request, Source}, 
    keywordstyle=[2]\bfseries\color{blue}, 
}

\definecolor{codegray}{rgb}{0.5,0.5,0.5}
\definecolor{codepurple}{rgb}{0.58,0,0.82}
\definecolor{backcolour}{rgb}{0.95,0.95,0.92}

\lstdefinestyle{mystyle}{
    backgroundcolor=\color{backcolour},   
    commentstyle=\color{codegray},
    keywordstyle=\color{blue},
    numberstyle=\tiny\color{codegray},
    stringstyle=\color{codepurple},
    basicstyle=\ttfamily\footnotesize,
    breakatwhitespace=false,         
    breaklines=true,                 
    captionpos=b,                    
    keepspaces=true,                 
    numbers=left,                    
    numbersep=5pt,                  
    showspaces=false,                
    showstringspaces=false,
    showtabs=false,                  
    tabsize=4
}

\newtheorem{remark}{Remark}

\usepackage[textsize=tiny]{todonotes}

\usepackage{amsmath}
\usepackage{amsfonts}
\usepackage{url}
\usepackage{xcolor}
\usepackage{enumerate}
\usepackage{booktabs}

\usepackage{comment}
\usepackage{mathtools}
\usepackage{braket}
\usepackage{multirow}
\usepackage{subcaption}
\usepackage{cleveref}

\usepackage{stmaryrd}
\usepackage{dblfloatfix}
\usepackage{accents}
\usepackage{makecell}
\usepackage{mathtools}

\newcommand{\Mahmoud}[1]{{\color{blue} \textbf{Mahmoud:} #1}}
\newcommand{\RM}[1]{{\color{red} \textbf{RM:} #1}}
\newcommand{\Ale}[1]{{\color{orange} \textbf{Alexandra:} #1}}

\def\reals{\mathbb{R}}

\def\nats{\mathbb{N}}
\def\prob{\mathbb{P}}

\def\M{\mathcal{M}}

\newcommand{\tmix}{\tau_{\textsf{mix}}}

\def\piinf{\pi_{ss}}

\def\timeout{{\boldsymbol{\tau}}}
\def\retry{{\boldsymbol{\rho}}}
\def\low{{\mathsf{Low}}}
\def\high{{\mathsf{High}}}

\def\simulator{{\mathsf{sim}}}
\def\tick{{\mathsf{T}}}
\def\dataset{{\mathcal D}}

\usepackage{tikz}
\usetikzlibrary{arrows.meta}

\begin{document}

\title{Formal Analysis of Metastable Failures in Software Systems}

\author{
		Peter Alvaro
	}
\affiliation{%
	\institution{UC Santa Cruz and AWS}
	\country{USA}
}
	\author{
		Rebecca Isaacs
	}
\affiliation{%
	\institution{AWS}
	\country{USA}
}
	\author{
		Rupak Majumdar
	}
\affiliation{%
	\institution{MPI-SWS and AWS}
	\country{Germany}
	}
	\author{
		Kiran-Kumar Muniswamy-Reddy
	}
\affiliation{%
	\institution{AWS}
	\country{USA}
}
	\author{
		Mahmoud Salamati
	}
\affiliation{%
	\institution{MPI-SWS}
	\country{Germany}
}
	\author{
		Sadegh Soudjani
	} 
\affiliation{%
	\institution{MPI-SWS and University of Birmingham}
	\country{Germany}
}


\renewcommand{\shortauthors}{Alvaro et al.}
\begin{abstract}
    Many large-scale software systems demonstrate \emph{metastable failures}.
In this class of failures, a stressor such as a temporary spike in workload causes the
system performance to drop and, 
subsequently, the system performance continues to remain low
even when the stressor is removed.
These failures have been reported by many large corporations and considered
to be a rare but catastrophic source of availability outages in cloud systems.

In this paper, we provide the mathematical foundations of metastability in request-response server systems.
We model such systems using a domain-specific langoogguage.
We show how to construct continuous-time Markov chains (CTMCs) that approximate the semantics of the programs through
modeling and data-driven calibration.
We use the structure of the CTMC models to provide a visualization of the \emph{qualitative} global behavior of the model. 
The visualization is a surprisingly effective way to identify system parameterizations that cause a system to show 
metastable behaviors.

We complement the qualitative analysis with \emph{quantitative} predictions.
We provide a formal notion of metastable behaviors based on escape probabilities,
and show that metastable behaviors are related to the eigenvalue structure of the CTMC.
Our characterization leads to algorithmic tools to predict recovery times in 
metastable models of server systems.

We have implemented our technique in a tool for the modeling and analysis of
server systems.
Through models inspired by failures in real request-response systems, we show
that our qualitative visual analysis captures and predicts
many instances of metastability that were observed in the field in a matter of milliseconds.
When we compute recovery times based on our algorithms,
we find, as predicted, the times increase rapidly as the system parameters approaches  
metastable modes in the dynamics.

In summary, we provide the formal foundations and first analytical tools for analyzing metastability in software systems.

\end{abstract}



	\maketitle

	\section{Introduction}
\label{sec:intro}

A \emph{metastable failure} in a distributed system is characterized by a temporary failure whose effect persists over time, 
even after the failure condition goes away \cite{bronson2021metastable,Huang2022systemmetastability}.
They manifest in the following way.
A system processes requests in a ``normal'' mode and maintains a high goodput (throughput of useful work).
A temporary rare ``trigger'' event, such as a spike in the workload or a capacity loss in the service, makes the system transition to a degraded mode with low goodput.
However, the  
system remains ``stuck'' in the degraded mode even when the spike or the capacity loss goes away: goodput remains low
for a much longer time scale than the trigger event.
Metastability is a rare source of failure in distributed systems, but a surprisingly common culprit in widely reported outages in cloud systems \cite{bronson2021metastable,AWS2015,AWS2014,AWS2021_persistent,AWS2024}.

A common example of metastable failures is a \emph{retry storm} at a server.
Retries are a mechanism in distributed systems to deal with failures: if a request is not responded to within a certain timeout, something went wrong and the client is advised to retry the request.
While retries are an excellent mechanism to mitigate transient failures,
in rare occasions, they may form a sustaining effect: the additional workload from retries prevents the system to respond to requests on time, thereby leading to further client-side retries that increases the workload.
In the worst case, the retry storm propagates to multiple services, leading to a collapse in availability.

Most research in metastability in software systems has been empirical, through the analysis of case studies of system outage. 
Practitioners have observed systems stuck after a spike and subsequent work amplification
and have developed best practices to avoid bad behaviors.
Researchers have reproduced metastable behaviors in workload testing and developed a taxonomy of
triggers, amplification, and cascades.
However, despite significant operational and empirical work, we still lack theoretical understanding and 
tool support for \emph{predicting and analyzing} metastable behaviors.
It is our goal in this paper to provide a theoretical foundation for metastability and corresponding tool support.

Our paper is part of an ongoing, larger, effort to understand metastable failures in hyperscalers, as outlined
in a recent workshop paper \cite{analyzing-metastable-faiures-hotos-2025}. 
We focus here on the formal aspects of the larger context.

\begin{figure}[t]
\centering
\begin{minipage}{0.38\textwidth}
\includegraphics[width=1\textwidth, height = 1\textwidth]{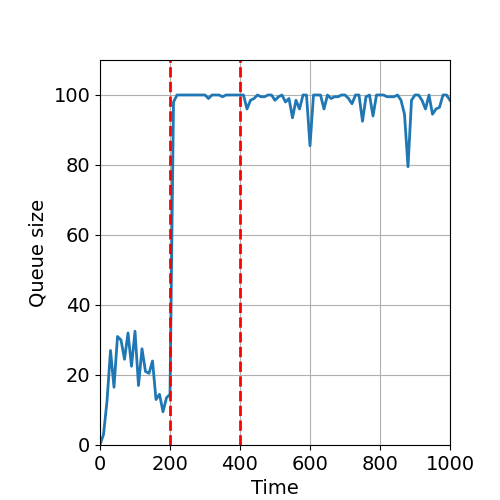}
\caption*{(a)}
\end{minipage}
\begin{minipage}{0.6\textwidth}
	\hspace{-0.4cm}\includegraphics[width=1.2\textwidth]{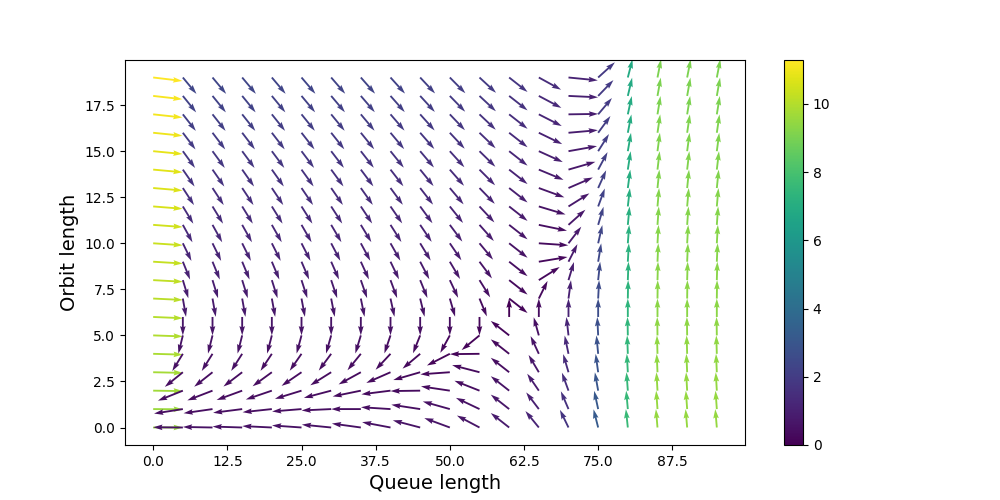}
    \caption*{(b)}
\end{minipage}

\caption{
A metastable failure and qualitative predictions from a formal model. 
(a) A simulated run of the example from \cref{fig:motivating-example}: the nominal arrival rate is 9.5 RPS. Between time 200s and 400s (between the red lines),
there is a load spike and the arrival rate is 20 RPS. 
The load goes back to the baseline at 400s.
In this simulation, the queue fills up but does not empty even after a further 600s after the load goes down. 
(b) A visualization of the stochastic dynamics of a CTMC model of the system. 
The arrows show the state change with the highest probability. The color of the arrow
represents the strength of the probability relative to the other transitions. 
}
\label{fig:intro}
\end{figure}

\paragraph{Motivating Example}
While metastable failures occur in many forms, we restrict ourselves to the setting of retry storms in request-response systems.
In a nutshell, these are systems in which clients send  requests that are handled by one or more servers.
Servers enqueue requests to absorb variabilities in the arrival rate.
Relatively rare events such as load bursts can cause queues to fill to such an extent that client requests time out and retry.
A failure occurs when there is a self-sustaining feedback loop of these client retries that prevents the system from performing any useful work.
Request-response systems are important components of cloud infrastructure---for example, low volume, critical operations like health checking or configuration updates are implemented as request-response systems---and retry storms are a common source of outages in these systems.

\begin{wrapfigure}{r}{5.7cm}
\centering
\hspace*{-0.2cm}
\includegraphics[scale = .5]{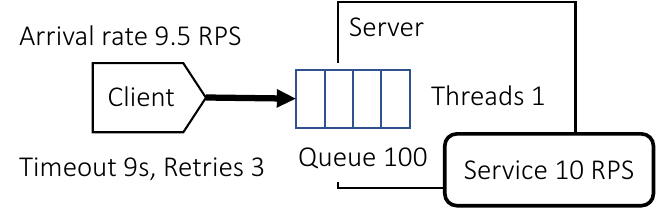}
\vspace*{-0.4cm}    
\caption{A simple example.}
\label{fig:motivating-example}
\end{wrapfigure}

As a canonical example of a retry storm, consider the following example (see \cref{fig:motivating-example}).
A system consists of a single server 
that serves requests with an exponential distribution with average rate of $\mu = 10$ requests per second (RPS). 
A client sends requests with an (independent) exponential distribution,
with average rate of $\lambda = 9.5$ RPS. 
Each request has a timeout of 9 seconds and retries 3 times before giving up.
When a temporary load spike fills up the queue to a size about 100, 
it is observed empirically that the queue does not drain and 
the failure rate of requests remains high for over 600s after the spike (see \cref{fig:intro}(a)).

The single-server system above is a ``classic'' example of a metastable failure: the queues remain full and the useful work done by the system remains near zero
long after the stressor is removed.
For this, and other examples of server systems, our goal
is to design models and mathematical analyses that explain
what goes wrong (\textbf{Q0}) much quicker than load testing.

Specifically, we aim to answer the following important questions from a service provider's perspective.
First, for what values of a system's parameters (queue sizes, arrival and service rates, retry policies) can metastable failures occur (\textbf{Q1})?
Second, can we predict the recovery time of a system after it has failed (\textbf{Q2})?
Third, can we provide predictions on the recovery time for common mitigations, such as throttling requests or autoscaling servers (\textbf{Q3})?

The ability to model and answer \textbf{Q0--3} are of enormous practical value: empirical load testing, as practiced today, is expensive (each test can take a day or more to set up and run).
Thus, it is simply infeasible to explore the parameter space or to make predictions about recovery.

\paragraph{Our work: Modeling and Analyzing Request-Response Systems} 
In this paper, we provide a formal lens to metastable failures in request-response services. 
We start with a domain-specific language (DSL) to model servers and clients, queues, requests, timeouts, and retries, with a discrete-event simulation (DES) semantics.
While, in principle, exhaustive simulations over the parameter space can answer \textbf{Q1--3} with statistical
guarantees, the cost of simulation is too high for such a strategy to be effective.
Instead, we consider an abstract model of the system that is amenable to more efficient algorithmic analysis.
%
%
Since the domain involves timing and probabilities,
we select \emph{continuous-time Markov chains} (CTMCs) \cite{harchol2013performance} as our modeling formalism.
CTMCs are state-transition models, in which the evolution of the state happens probabilistically in continuous time.
In each state, the CTMC waits for some duration of time, drawn from an exponential distribution, before transitioning to a neighboring state.

Our first contribution is to construct abstract CTMC models for request-response systems in the DSL, following
insights from retrial queueing systems \cite{Artalejo2008,habibi2023msfmodel}.
This is quite nontrivial: the DES maintains a large amount of state (queued requests, timers, timeout handlers)
and some features are not Markovian (timeouts and retries).
We abstract the simulator state into 
the size of each queue (modeling the number of requests
in the system---either being served or waiting in 
queues---at a point in time) and 
the \emph{orbit} (modeling the average effect of requests being retried).
The transitions of the CTMC abstract away the operational
details of the simulator, and only consider the \emph{average} arrival rate, service rates, 
and retry rates.

The CTMC model abstracts away many details of actual systems, but allows us to make qualitative and quantitative
predictions about \textbf{Q1--3}.
However, a consequence of the abstraction is that the predictions of the model can deviate significantly from the operational behavior of the simulator.
Therefore, as a second step, we perform \emph{data-driven calibration}
of the CTMC model using simulation data.
We consider short simulation runs of the system, and use
these runs to calibrate the parameters of the abstract CTMC
to ensure that the trajectories of the CTMC agrees with
the simulator with respect to the CTMC state.

The \emph{ab initio} formal modeling and the calibration are synergistic: pure formal modeling deviates from data,
but learning from simulation traces without any prior structure performs poorly as well.
%
%
We show empirically that the calibration is crucial in
obtaining precise quantitative predictions of real systems.
In particular, we show that our calibrated model can predict, statically, metastable behaviors in our simple example for 
different parameter ranges (partly answering \textbf{Q0}).

%




%

%
\begin{wrapfigure}{r}{4.5cm}
\centering
\vspace*{-0.1cm}
\hspace*{-5mm}
\includegraphics[clip, trim=0.9cm 0.9cm 0.9cm 0.9cm, scale = 0.15]{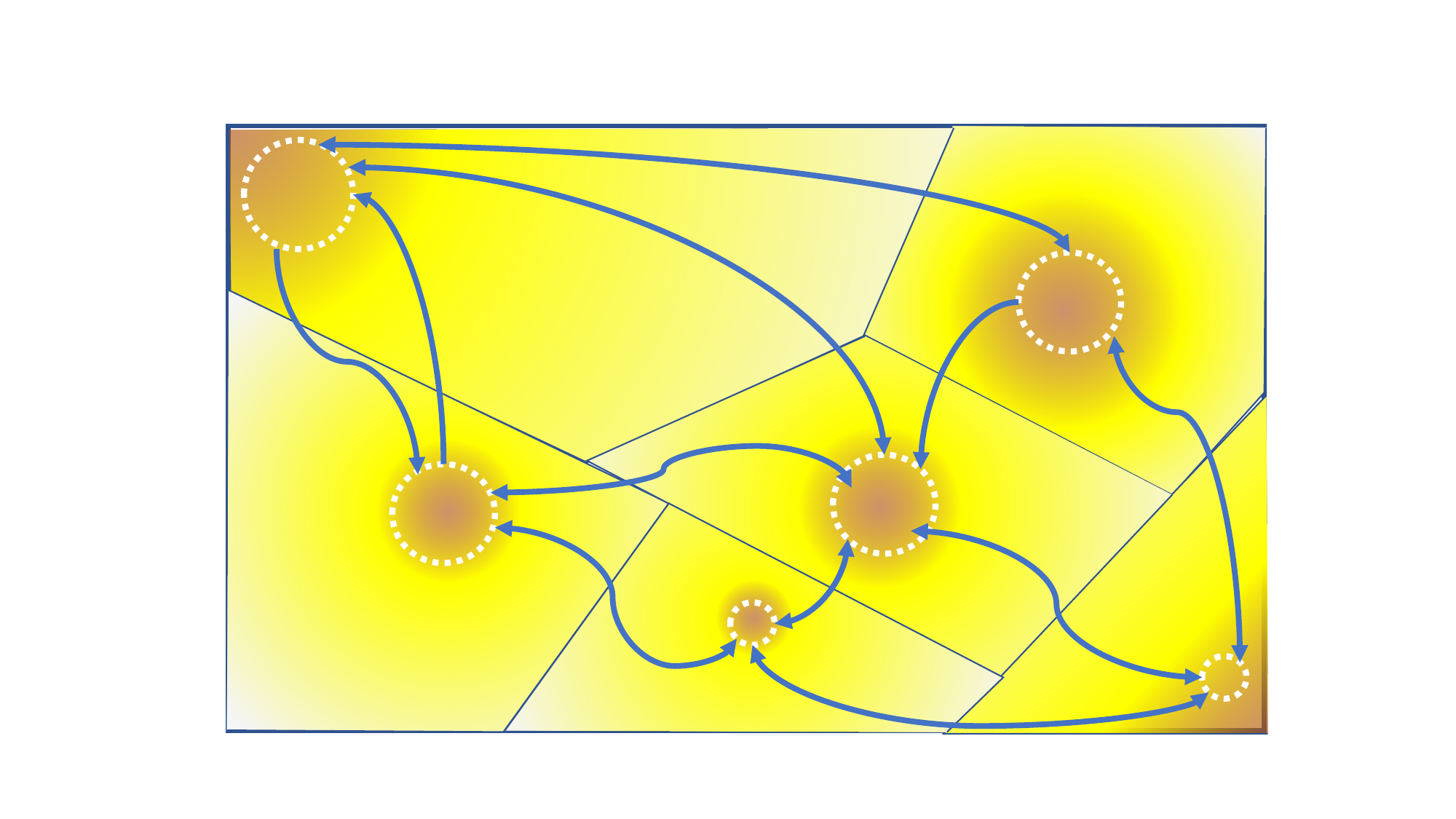}
\vspace*{-0.3cm}
\caption{Metastability, pictorially.}
\label{fig:metastable}
\end{wrapfigure}
\paragraph{What is Metastability?}
Having a fixed mathematical model (CTMCs), we focus on a formal characterization of metastable behaviors. 
Intuitively, metastability corresponds to the existence of two or more well-separated time scales, such that the system remains
in an ``almost invariant'' set in the short time scale but can visit different almost invariant sets in the long time scale.
Figure~\ref{fig:metastable} gives a visual depiction of the time scales: the orange balls 
denote almost invariant sets.
Once entered, the CTMC remains there for a long time but can move to a different almost invariant set over a long time horizon.
States outside the white balls enter one of the almost invariant sets over a short time scale.

%


Following the results in the theory of stochastic dynamical systems \cite{FW,Betz2016,BEGK2002,MetastabilityBook},
we make this intuition precise by defining metastability in a CTMC using \emph{escape probabilities}.
We show that the notion is robust by providing alternate
characterizations using the eigenvalues of the CTMC, thus answering
\textbf{Q0}.

Surprisingly, existing literature on probabilistic verification does not consider metastability as a temporal specification.
Probabilistic temporal logics focus either on transient
behaviors or stationary behaviors of the system \cite{BaierHHK03,BaierKH99,CY95,BdA95,BaierKatoen,Prism,Storm}.
Metastable behaviors provide a finer structure on the time-evolution of the system, and are not captured by logics such as CSL \cite{Aziz:2000} or probabilistic linear temporal logics \cite{CY95}.

Similarly, classical queueing theory \cite{harchol2013performance}
focuses on stability or instability of a system.
In fact, many standard queueing models, including the M/M/c queue, does not exhibit metastability!
A recent attempt \cite{habibi2023msfmodel} defined metastability as a large expected distance to the origin.
Unfortunately, this definition conflates metastable behaviors with unstable ones: an unstable system satisfies the definition but metastable behaviors occur in stable queues as well.

%
%

\paragraph{Qualitative Predictions}
While the notion of metastability is defined for any CTMC, we show a simple visualization for request-response systems.
Since the state of each server is two dimensional (its queue and its orbit), and the CTMC is sparse (each transition goes to a neighboring state),
there is a two-dimensional plot of the stochastic dynamics.
Our visualization captures the dominant direction of flow in the stochastic dynamics defined by the CTMC.
We have found that the visualization captures the \emph{qualitative} phenomenon of metastability in the global parameter space.
Moreover, even when the CTMC has many states, the visualization---which involves computing the transitions of a small number of states---can be produced in milliseconds for each server!
In contrast,
anecdotal evidence suggests that---even for simple queueing models---reproducing metastable failures by careful parameter selection required heroic effort.
This lets us answer \textbf{Q1}.

As an example, Figure~\ref{fig:intro}(b) 
shows the visualization for our running example. 
The x-axis is the queue length, and the y-axis is the orbit length.
The direction of each vector points to the most probable next state; the relative magnitude of the probability of moving in this direction is given by the color of the arrow.
Arrows to the left and to the bottom ``clear out'' queues and retrying requests; arrows to the right and to the top increase the queue and the rate of retries.
We can visually see the point of metastability: at a queue length of about 100, 
if there are enough retrials in the system, the system changes its qualitative dynamics.
Beyond this point, filled queues are likely to remain full.
These observations correspond to our intuition: the average latency of requests exceeds the timeout at this point, triggering retries and moving the dynamics ``away'' from a small queue.

\paragraph{Quantitative Predictions}
The visualization is backed up by quantitative predictions from the underlying CTMC.
In request-response systems, an important question is to quantify the \emph{recovery time} (the time taken for a system to go from a full queue (e.g., after a load spike)
back to the average queue (average queue size in the stationary distribution)), as well as the recovery time after adding throttling (see \cref{fig:quant}).
Using standard algorithms for CTMCs, we can compute recovery times---either exactly by solving a linear system or 
approximately through estimates of eigenvalues---and we show how metastable
regions in the visualization correspond to very large 
expected recovery times.
This lets us answer \textbf{Q2-3}.
We note that calibrating the model is essential to finding good
quantitative predictions.

\paragraph{Implementation}
We have implemented our analysis for metastability in server systems in an open source tool.
Our implementation provides a flow from the DSL and its simulator to a CTMC as well as visualization and analysis tools on the CTMC.
We show by a number of experiments that our analysis is able to explain, reproduce, or predict metastable failures in models of request-response systems.
Moreover, the qualitative analysis runs in milliseconds,
and the quantitative analysis runs in a few hours even for our largest examples.

In applying our tools to industrial examples in a hyperscaler, we have found that modeling a small number ($\leq 3$) of servers and queues is sufficient to reproduce many metastability issues.
In our experience, the predictions of the CTMC models 
allow us to find and to reproduce metastable effects within a few hours, rather than weeks.

The abstract CTMC models do not capture the system with all fidelity, and we still rely on the simulation (and emulation) to check predictions or perform further performance analyses. 
However, despite the abstraction, in an industrial context, we have found the abstraction and analysis indispensable to  find where to focus our efforts for simulation and workload testing.
Since workload testing of services is expensive, the abstract modeling can substantially reduce the required testing efforts.

\paragraph{Contributions}
We make the following contributions in this paper.
\begin{enumerate}
\item We formalize request-response systems in a DSL and show
how the simulation-based semantics of the DSL can be approximated by an abstract CTMC.
We provide a methodology that combines formal modeling with data-driven calibration to ensure accuracy of predictions.

\item We provide a formal foundation to metastable failures in software systems in terms of metastable states in CTMC models.
We define metastable states based on escape probabilities and also give a spectral characterization.
\item 
We show that our CTMC models 
provide qualitative (visual) 
information that predicts global parameters that lead to metastable behaviors.
The CTMC models also provide quantitative predictions about recovery times.
\item We show that our algorithms can be used to find metastable failures on models of real request-response systems.
\end{enumerate}



        \section{Overview}
\label{sec:overview}

\paragraph{Motivating example: Modeling} 
Let us come back to the motivating example from \cref{fig:motivating-example}.
Our goal is to show how we model it as a continuous-time Markov chain (CTMC) and what analyses we can perform.

The states of the CTMC will be pairs of integers $(u, v)$.
The first index $u$ tracks the number of requests in the system (either
being served by the server or waiting in the queue) 
and the second index $v$ tracks the number of requests in the ``orbit'' that failed and are currently waiting to be retried.
The first index tracks the workload in the system, the second tracks the work amplification due to retries.

Transitions between states are determined by two factors:
(1) the arrival rate from the client and the service rates, 
(2) the timeout and the retry policy.
An arriving request increases the number of requests in a queue from $u$ to $u+1$ with rate equal to the rate of request arrivals.
Finishing a request reduces the number of requests in the system, so the number of requests go from $u+1$ to $u$; this happens
at the service rate.
In addition, the number of requests can increase by putting a request from the orbit into the queue and the number of requests in
the orbit can increase at a rate determined by a rate computed from the arrival rate and the timeout. 
For each state, we can break the transition probabilities into two components: the first tracks the change in the queue axis
and the second tracks the change in the orbit axis. 
The \emph{nominal model} will use the parameters from the example---for example, the arrival rate will be 9.5, the service rate will be 10, and the other probabilities will be determined from the constants in the program.

However, we will calibrate the nominal transition rates using simulation data, ensuring that the trajectories produced by the calibrated CTMC closely align with those of the simulator.
In our case, the calibrated CTMC has parameters $\lambda' = 9.43$ and 
timeout $10.54$s.

%

\begin{figure}[t]
\begin{minipage}{0.6\textwidth}	\hspace{-0.4cm}\includegraphics[width=1.2\textwidth]{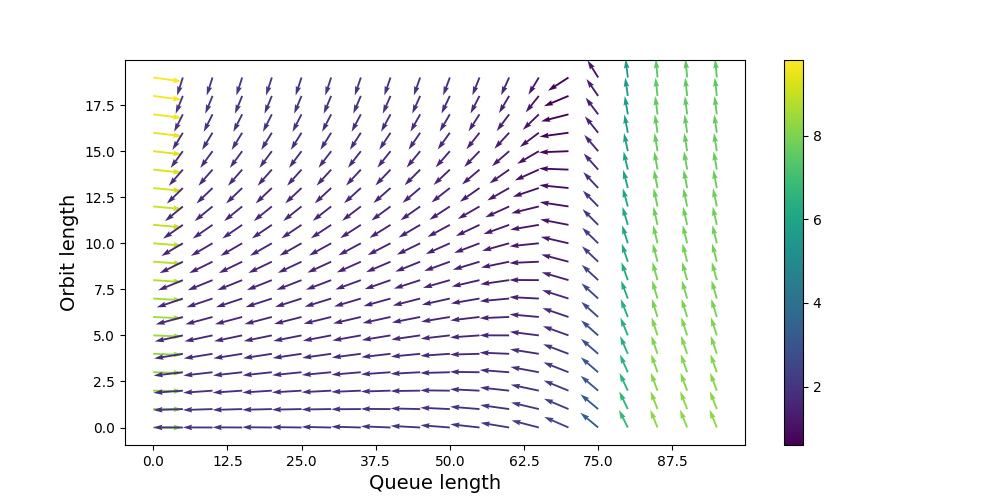}
    \caption*{(a)}
\end{minipage}
\begin{minipage}{0.38\textwidth}
\includegraphics[width=1\textwidth, height = 1\textwidth]{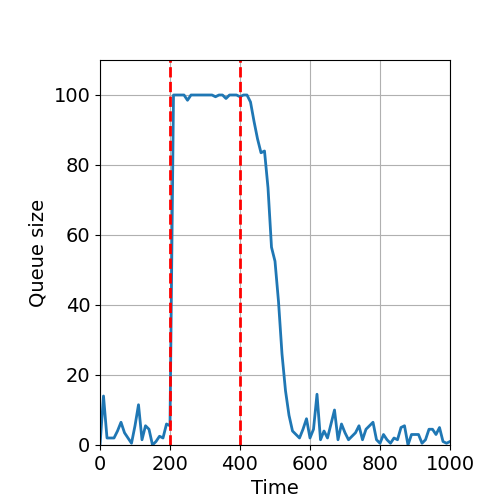}
\caption*{(b)}
\end{minipage}~
\caption{Throttling a system after a load spike to recover quickly.
(a) A visualization of the stochastic dynamics when the arrival rate is throttled to 8 RPS.
(b) A simulated run that confirms quick recovery. 
}
\label{fig:example_plots}
\end{figure}

\begin{figure}[t]
\begin{minipage}{0.38\textwidth}
\includegraphics[width=1\textwidth, height = .95\textwidth]{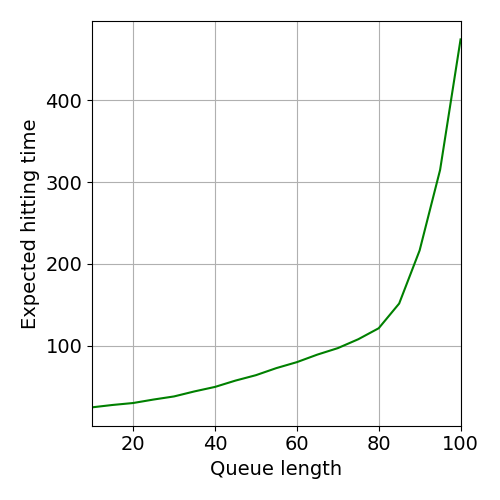}\vspace{-1.5em}
\caption*{(a)}
\end{minipage}%
\begin{minipage}{0.38\textwidth}
\includegraphics[width=1\textwidth, height = .95\textwidth]{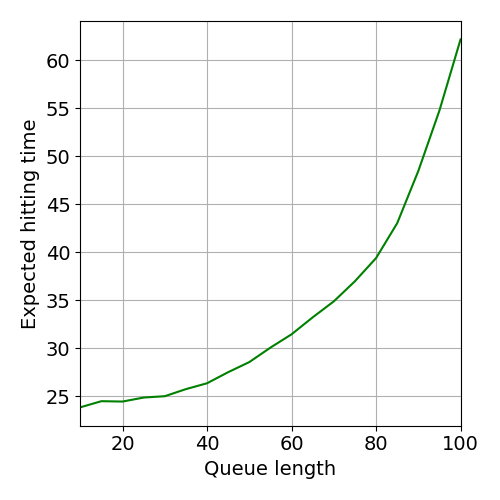}\vspace{-1.5em}
\caption*{(b)}
\end{minipage}~
\caption{Quantitative predictions of the time for a full queue to return to its average size pre-load spike, as a function of queue length.
(a) For the metastable case ($\lambda=9.5$ RPS), (b) for the throttled case ($\lambda = 8$ RPS).}
\label{fig:quant}
\end{figure}

\paragraph{Warm-up: No retries}
Let us first consider the special case in which requests are not retried and the state is only one dimensional. 
This special case corresponds to the classical model of M/M/1 queues: requests arrive 
with rate $\lambda=9.5$, they are served with rate $\mu=10$.
The \emph{qualitative} dynamics of the CTMC has two forms:
if $\lambda < \mu$,  
the transition rate to the left will dominate the rate to the right,
and conversely, if $\lambda > \mu$.
(As an exception, at the state 0, the arrows will always point right, since a new arrival will be immediately served.)
The following figure visualizes the dynamics in these two cases,
arrows point right if arrivals outweigh service times, and left otherwise:

\vspace*{3mm}
\begin{tikzpicture}[
    tbox/.style={font=\scriptsize, minimum width=1.5cm, minimum height=1cm, align=center},
    node/.style={circle, draw, minimum size=0.15cm, inner sep=0},
    below id/.style={font=\scriptsize, below=1mm},
    edge/.style={->, thick},
    dotted edge/.style={->, dotted, thick}
]

\node[tbox] (t) at (-0.8, 0.5) {$\lambda > \mu$:};
\node[node] (A) at (0, 0.5) {};
\node[below id] at (A.south) {0};

\node[node] (B) at (1, 0.5) {};
\node[below id] at (B.south) {1};

\node[node] (C) at (2, 0.5) {};
\node[below id] at (C.south) {2};

\node[node] (D) at (3.5, 0.5) {};
\node[below id] at (D.south) {N};

\draw[edge] (A) -- (B);
\draw[edge] (B) -- (C);
\draw[dotted edge] (C) -- (D);

\node[tbox] (t) at (5.2, 0.5) {$\lambda < \mu$:};
\node[node] (AA) at (6, 0.5) {};
\node[below id] at (AA.south) {0};

\node[node] (BB) at (7, 0.5) {};
\node[below id] at (BB.south) {1};

\node[node] (CC) at (8, 0.5) {};
\node[below id] at (CC.south) {2};

\node[node] (DD) at (9.5, 0.5) {};
\node[below id] at (DD.south) {N};

\draw[edge] (AA) -- (BB);
\draw[edge] (CC) -- (BB);
\draw[dotted edge] (DD) -- (CC);

\end{tikzpicture}

Intuitively, these cases correspond to random walks on the line with a drift to the left or to the right.
When the arrival and service rates match exactly, the two directions balance out; this corresponds to a random walk with equal probability
to move left or right.

Classical results on random walks back up the visual analysis with quantitative results.
In the first case, a full  queue will remain full 
(the queueing system is called \emph{unstable}), and queues may drain with
exponentially small probability.
In the second case, a small queue will tend to remain small (the system is \emph{stable}), and
a full queue will empty out in time linear in the size of the queue. 

Interestingly, there is no metastable behavior in M/M/1 queues; the behaviors are stable or unstable, based on the two cases.
This partly explains why a mathematical study of metastable modes is conspicuously absent in the queueing and probabilistic modeling literatures.

\paragraph{Visualization of retries and metastability}
Let us return to our example. 
The dynamics are richer when retries are involved.
Figure~\ref{fig:intro}(b) shows a visualization of the dynamics of the CTMC with retries.
Each arrow in the figure represents a normalized vector, whose direction provides the \emph{most probable} relative change in the queue and orbit axes,
and whose magnitude provides the normalized rate of transitions.
When queues are small and the orbit has few requests, the 
arrows point ``downward'' and ``leftward''. 
Thus, the queue and orbit clear out with high probability. 
At a critical point---around a queue size of 65---when
many requests time out, 
the dynamics ``drifts up and right'', causing long queues and retries to amplify. 
This marks a point of metastability, where the system keeps queues full due to retries.

From the visualization, one expects that the time to recover from a full queue to the average queue  
takes a sharp turn
as the queue length increases and the dynamics moves to a metastable regime.

The visualization enables us to answer qualitative questions about metastable behaviors across the global space of configuration parameters (answering \textbf{Q1}).
As we show below, the qualitative intuition can be confirmed quantitatively.

\paragraph{Quantitative analysis}
%
Quantitative  analysis of the CTMC  backs up the intuition from the visualization. 
For example, we can compute the \emph{expected hitting time}.
Figure~\ref{fig:quant}(a) shows that the expected hitting time from a full queue (90\% or more of the queue length) to the average queue
increases with the queue length, with a sharp increase around queue length of 75, the metastable point in the visualization, and increases rapidly as the queue length goes beyond that.
The underlying analysis is based on solving linear systems of equations and is implemented using efficient numerical linear algebra routines---for this model, the run time is a few minutes.

Thus, in addition to intuition about the global parameter space, the CTMC model allows us to make quantitative predictions about system recovery (answering \textbf{Q2}).

\paragraph{Effect of recovery strategies}
Finally, we can use the same analyses to answer \textbf{Q3}: 
what is the effect of a recovery policy on recovery time?
In practice, one way to recover a system is to throttle the incoming requests to a lower value, so that the queues can clear.
Figure~\ref{fig:example_plots}(a) shows a visualization of the dynamics when the arrival rate is throttled at 8 RPS.
The dynamics drifts down and left, so this is a good choice for throttling the input.
Figure \ref{fig:quant}(b) predicts that the recovery time for this 
throttled arrival rate should be low.
Figure~\ref{fig:example_plots}(b) confirms through a simulation run that the system recovers quickly when the arrival rate is throttled to 8 RPS after the load spike. 

\begin{figure}[t]
\centering
\vspace*{-0cm}
\hspace*{0mm}\includegraphics[scale=0.5]{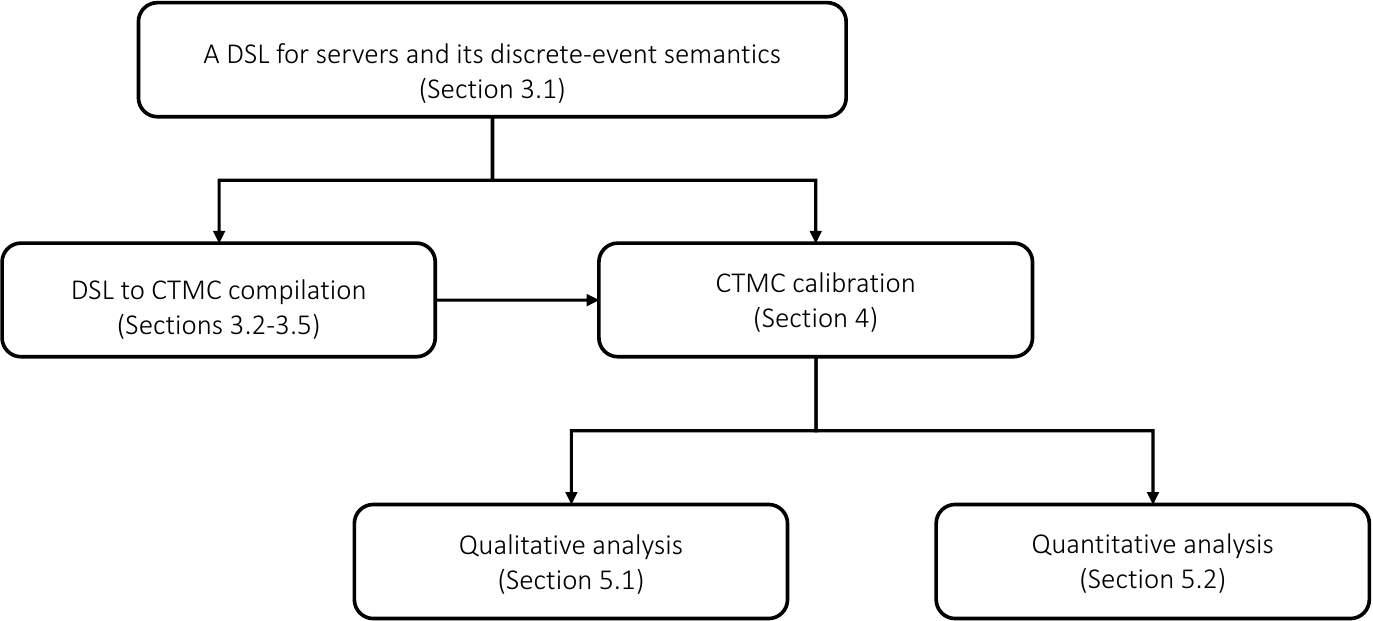}
\caption{Overall scheme for analyzing metastability in server systems.}
\vspace*{-0.0cm}
\label{fig:overall_diagram}
\end{figure}

\paragraph{Outline} 
In the remainder of this paper, we provide a detailed description of different components depicted in \cref{fig:overall_diagram}. Section~\ref{sec:ctmc_models} introduces the syntax and semantics of the DSL for specifying server systems, along with its compilation into CTMCs. Section~\ref{sec:ctmc_calibration} explains how simulation trajectories can be leveraged to calibrate the \emph{ab initio} CTMCs obtained from direct DSL compilation. In Section~\ref{sec:analysis}, we present both qualitative and quantitative analyses aimed at characterizing metastability in server systems. Finally, Section~\ref{sec:experiments} reports experimental results that demonstrate the effectiveness of our analysis, and Section~\ref{sec:conclusion} concludes the paper.

\section{Modeling Request-Response Systems as CTMCs}\label{sec:ctmc_models}

\subsection{A Domain-Specific Language for Systems}
\label{sec:dsl}

We express request-response systems in a simple DSL
(embedded in Python)
that provides abstractions for servers and clients.
A \emph{server} maintains a queue of requests and a pool of workers.
The workers pull requests off the queue and process them asynchronously.
Processing a request can incur a delay determined by the service time distribution for that request type.
Moreover, processing a request may make further calls
to downstream servers.
In our model, the worker processing a request blocks
until the downstream calls return.

\emph{Clients} send requests to the server.
Clients generate new requests based on an arrival distribution.
Each request has a timeout 
as well as a retry policy (e.g., number of retries, backoff).
Clients enqueue their request on a server.
If an enqueued request times out, the client may send further
attempts to the server based on the retry policy.

Programs in the DSL, such as the simple example in \cref{fig:motivating-example},
are acyclic graphs connecting clients and servers to other servers.
The semantics of a program is given by a discrete-event simulation.
\cref{fig:simulation-semantics} shows the core of the simulator.\footnote{
    It is easy to give a formal operational semantics for the language. The operational semantics maintains timestamped requests in the system and updates the state based on a global timer. Instead, we provide the code to show the simplicity of implementing the semantics: the core is about a 100 lines of Python but already provides an effective simulation model for real systems!
}
We treat the discrete-event simulation as the ground truth when comparing the predictions of the CTMC models.

\begin{figure}[t]
\begin{minipage}[t]{0.48\linewidth}
\begin{lstlisting}[language=Python, basicstyle=\tiny]
class Server:
  ...
  async def enqueue(self, request: Attempt):
    if self.queue.full(request):
      request.future.set_result(f"Dropped")
    else:
      await self.queue.put(request)

  async def worker(self, worker_id):
    while self.running:
      request = await self.queue.get()

      // the simplest processing is to delay,
      // but we can make downstream calls
      await asyncio.sleep(
        self.service_time_dist.sample(request))
      if not request.future.done():
        request.future.set_result(
              f"Success {worker_id}")
      self.queue.task_done()
\end{lstlisting}
\end{minipage}
\begin{minipage}[t]{0.49\linewidth}
\begin{lstlisting}[language=Python, basicstyle=\tiny]
class Client:
  ...
  async def send_request(self):
    request = Request(self.id, reqtype=..., arg=...)
    for attempt in range(1, self.retries + 1):
      resp = asyncio.get_event_loop().create_future() 
      req = Attempt(request, resp) 
      await self.server.enqueue(req)
      try:
        ret = await asyncio.wait_for(
            asyncio.shield(resp), timeout)
        return
      except asyncio.TimeoutError:
        await self.retry_policy()

  async def run(self):
    while ...:
      await asyncio.sleep(self.arrival_dist.sample())
      task = async.create_task(self.send_request())
      ...
\end{lstlisting}
\end{minipage}
\caption{Simulator implementation. The simulator gives an operational semantics to the DSL. We use Python's \emph{asyncio} library. \emph{async} denotes an asynchronous call (a future), \emph{await} waits for an asynchronous call to finish. \emph{sleep} blocks until some time has passed. \emph{tasks} are run on a separate thread and does not block the main thread; \emph{wait\_for} waits for an asynchronous task to finish, \emph{shield} ensures tasks are not cancelled. Internally, the async runtime maintains state in the form of requests, futures, and timers.
}
\label{fig:simulation-semantics}
\end{figure}

\subsection{Continuous-time Markov Chains (CTMCs)}

Our goal is to ``compile'' programs in the DSL as CTMCs,
such that the behavior of the CTMC matches the simulation semantics.
We assume familiarity with the basic theory of CTMCs (see, e.g., \cite{Norris,anderson2012continuous}) but provide a recap of basic definitions.

A \textit{continuous-time Markov chain} (CTMC) is a stochastic process over a discrete state space.
The process makes transitions from state to state, independent of the past.
Upon entering a state, it remains in the state
for an exponentially distributed amount of time before changing its state.
This time is called the \textit{holding time} at the state.

Formally, a CTMC $\M = (S,Q)$ 
consists of a set $S$ of \emph{states} and a generator matrix $Q$. 
The generator matrix satisfies
\[
Q_{ii} =  - \sum_{j \neq i} Q_{ij}.
\]
Intuitively, $Q_{ij}>0$ for $i\neq j$ indicates that a transition from $i$ to $j$ is possible and that the timing of the transition is exponentially distributed with rate $Q_{ij}$. 

The probability distribution of a 
CTMC $\M=(S,Q)$ is a continuous function of time that 
evolves according to the forward Chapman-Kolmogorov differential equation:
\begin{equation}
	\label{eq:Kolmog}
	\frac{d}{dt} {\pi}(t) = 
	\pi(t) {Q}, \quad {\pi}(0) = \pi_0,
\end{equation}
where $\pi_0\in[0,1]^{|S|}$ denotes the initial distribution mapping states to probabilities. 
The unique solution to the equation is given by
$\pi(t)  = \pi_0 e^{Qt}$,
where $e^{Qt}$ is the matrix exponential function.

Consider a CTMC $\M = (S,Q)$ and denote its state at time $t$ by $X(t)$.
Letting $t_n$ denote the time at which the $n^{\text{th}}$ change of state (transition) occurs, we see that
$X_n = X(t_n^+)$, the state right after the $n^{\text{th}}$ transition, defines an underlying discrete-time
Markov chain, called the \emph{embedded Markov chain}. $X_n$ keeps track of
the states visited right after each transition, and moves from state to state according to
the one-step transition probabilities $P_{ij} = \prob(X_{n+1} = j|X_n = i)$. 

\subsection{Basic Models: M/M/1 Queues and No Retries as CTMCs}
 
\vspace*{.25cm}
\begin{wrapfigure}{r}{5.5cm}
\vspace*{0cm}
\hspace*{0mm}\includegraphics[scale=0.45]{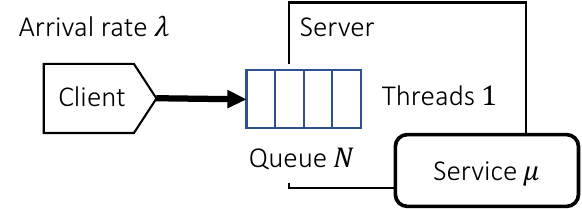}
\caption{An M/M/1 queue in the DSL.}
\vspace*{-0cm}
\label{fig:mm1_diagram}
\end{wrapfigure}
An M/M/1 queue models a simple client-server system with a single First In, First Out (FIFO) queue and a single server.
Clients send requests according to an exponential distribution with rate $\lambda$.
Requests are enqueued at the tail and processed in FIFO order.
Requests have service times that are exponentially distributed at rate $\mu$. The service times are independent from each other and from the arrival process.

%

M/M/1 queues are modeled as CTMCs \cite{Norris}.
The states of the CTMC correspond to the number of requests in the system (either being processed or in the queue).
Arrivals increase the number of requests at rate $\lambda$,
served requests decrease it at rate $\mu$.

The transition probabilities $P_{ij}$ for the embedded discrete-time chain are as follows.
If $X_n = 0$, then we are waiting for an arrival, so $\prob(X_{n+1} = 1\mid X_n = 0) = 1$.
If $X_n = i$ for some $i\geq 1$, then $X_{n+1} = i+1$ with probability $\prob(X < S_r) = \lambda/(\lambda + \mu)$ and $X_{n+1}= i- 1$ with probability $\prob(X > S_r) = \mu/(\lambda + \mu)$,
depending on whether an arrival or departure is the first event to occur next.
Thus, the embedded Markov chain is a simple random walk with ``up'' probability $\lambda/(\lambda + \mu)$ and ``down'' probability $\mu/(\lambda + \mu)$, that is restricted to be non-negative
$P_{0,1} = 1$.

If the number of requests is bounded to $N$ elements, and a request that arrives when the queue is full is lost, we can
modify the CTMC as follows.
The state space is $\set{0, \ldots, N}$.
The transition function now enforces that $P_{N,N - 1} = \mu/(\lambda+\mu)$ and $P_{N, N} = \lambda/(\lambda + \mu)$, i.e., arrivals when the queue is full are dropped.

\subsection{Modeling a Single Server and Clients with Timeouts and Retries}\label{subsec:single_server_modelling}

We move on to model timeouts and retries.
\emph{Timeout} means that there is a constant $\timeout$ such that, if a request has not been served within $\timeout$, a client can take further action. 
This can be a \emph{retry}, where a new instance of the service is enqueued (without removing the original instance), or a \emph{drop}, where the client decides to drop the request. 
In order to model the effect of retries, we augment the states of the CTMC to track not only the requests in the system (being processed at the server or waiting in the queue),
but also an \emph{orbit} in which requests wait to be retried \cite{Artalejo2008,habibi2023msfmodel}.

We provide the compilation step-by-step,
starting with a simple case and adding 
more features to the model,
without polluting the central intuitions.

\begin{wrapfigure}{r}{5.5cm}
\vspace*{-0.5cm}
\hspace*{0mm}\includegraphics[scale=0.45]{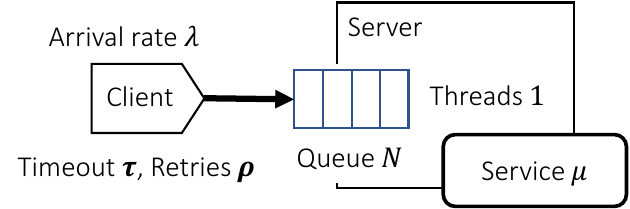}
\caption{A retrial queue in the DSL.}
\vspace*{-0.0cm}
\label{fig:retrial_queue_diagram}
\end{wrapfigure}

\paragraph{One Server, One Thread, Requests with Timeouts and Retries}
A state of the CTMC is a pair $(u, v)$, where $u\in \nats$ is the number of requests in the system (being processed in the server or waiting in the queue),
and $v\in \nats$ is the number of requests waiting in the orbit to be retried. 
We assume that there is a fixed timeout $\timeout$ for all requests and requests that time out are retried up to $\retry$ times.

The CTMC models the following processes, all mutually independent:
\begin{itemize}
\item Requests arrive according to a Poisson distribution $\{\textcolor{red}{A(t)}\colon t\geq 0\}$ with rate $\lambda$; 
\item Processing time at the server is a process $\{\textcolor{blue}{C(t)}\colon t\geq 0\}$ which is a Poisson distribution with rate $\mu$; 
\item Failures $\{\textcolor{brown}{F(t)}\colon t\geq 0\}$ correspond to the event that an incoming request will not be served within the timeout horizon and hence added to the orbit: 
if the current state is $(u, v)$, the probability of such an event can be calculated as $r(u)\coloneqq \sum_{i=1}^u \frac{(\mu\timeout)^i}{i!}e^{-\mu\timeout}$;
\item Retries $\{\textcolor{magenta}{R(t)}\colon t\geq 0\}$ which brings the requests waiting in the orbit into the queue: 
if the current state is $(u, v)$, the corresponding transition happens with the (average) rate $ \frac{\retry v}{(\retry +1)\timeout}\coloneqq\alpha v/\timeout$;
\item
        Drops $\{\textcolor{green}{D(t)}\colon t\geq 0\}$ correspond to requests that have been retried $\retry$ times and therefore will be abandoned:
if the current state is $(u, v)$, we consider an exponentially distributed sequence with rate $\frac{v}{(\retry +1)\timeout}\coloneqq (1-\alpha)v/\timeout$. 
\end{itemize}

We note that the processes related to exogeneous arrivals and retries involve adding new requests to the queue.
The rate is dependent on the current number of jobs in the queue:
for a queue size $u$, it depends on the failure rate $r(u)$ whether a~new request will also be added to the orbit or not. 

Formally, the CTMC $\M$ has the state space $\nats^2$.
If the current state is $(u, v)$, the transition rates are defined as follows:
\begin{itemize}
\item both exogenous arrival and (predicted) timeout: $Q((u,v), (u+1, v+1)) = \lambda r(u)$,
\item exogenous arrival but no timeout: $Q((u,v), (u+1, v)) = \lambda (1 - r(u))$,
\item request completion: $Q((u,v), (u-1, v)) = \mu$, if $u \geq 1$,
\item queue a request to be retried, but assume it will fail: $Q((u,v), (u+1, v)) = \alpha v r(u)/\timeout$,
\item queue a request to be retried and assume it will succeed: $Q((u,v), (u+1, v-1)) = \alpha v (1- r(u))/\timeout$,
\item drop a  request from the orbit: $Q((u,v), (u, v-1)) = (1 - \alpha)/\timeout v$.
\end{itemize}
%

As before, if the queue is bounded by $N$,
we modify the transition rules to ensure 
$u \leq N$ on every transition by disabling
transitions that increment $u$ when $u = N$.
The effect is that when the queue is full, 
new requests are dropped.

\paragraph{Thread Pools with Multiple Threads.}
When a server has multiple threads, we generalize the CTMC models for M/M/1 queues.
Suppose there are $c$ threads.
When $u < c$, some threads are free to serve arriving requests, and the transition
rates are determined by the competition between an arrival
and the completion of the $u$ threads.
When $u \geq c$, the holding time is determined by the arrival rate as well as the (independent) competing service times of each thread.

Thus, focusing only on the number of jobs in the queue, $P_{0,1} = 1$ and for $0\leq i < c$, $P_{i,i+1} = \lambda/(\lambda + i \mu)$, $P_{i, i-1} = i\mu/(\lambda + i\mu)$.
For $i \geq c$, $P_{i, i+1} = \lambda/(\lambda + c\mu)$ and $P_{i, i-1} = c\mu/(\lambda + c\mu)$.

\paragraph{Multiple Request Types}
In general, a server accepts multiple request types,
each with their own service rates.
Multiple clients can connect to a server, each with their
own arrival rates.

In the CTMC, we model a single queue with all request types.
We average the request arrival rates over all clients
and model that \emph{on average} we will get a request with arrival rate $\lambda_i$ with probability
$\lambda_i / \sum_j \lambda_j$.
Similarly, we average the service rate with
the average of the individual service rates with this arrival distribution. 

\begin{remark}
The CTMC model ``averages out'' 
the simulation semantics. 
The simulator state maintains individual request attempts
and may have multiple outstanding attempts for the same 
request using \texttt{Timeout}s to generate new attempts
while retries remain.
Instead, the CTMC captures the average behavior of the requests: 
\emph{on average}, $r(u)$ fraction of requests time out, \emph{on average}, requests are retried with rate $\alpha v/\timeout$ and dropped with rate $(1-\alpha)v/\timeout$, and so on.
We recover the fidelity of the approximations using
data-driven calibration (\cref{sec:ctmc_calibration}).
\end{remark}

        \subsection{Multiple Servers}\label{App:multi_server_model}

Finally, we consider multiple servers.
A request at a server can be sequentially forwarded
to downstream servers.
A request is considered served when it is processed
by a leaf server.
We only consider the \emph{synchronous} mode,
where upstream server threads remain blocked
until the request is served---many real-world 
request-response systems are implemented on top of
synchronous Remote Procedure Call (RPC) 
infrastructure.
For notational simplicity, we describe the construction
when the servers are pipelined (\cref{fig:multi-server-sketch}) and there is only one request type.
The ideas carry over to more general acyclic graphs.\footnote{
    In queueing theory, re-entrant queues form cyclic
    graph structures. We have not seen such configurations
    in request-response systems in the cloud context.
}


We fix a program with $K$ servers.
Each server $i\in \set{1,\ldots, K}$ 
is attached to a client with arrival rate 
$\lambda_i$, timeout $\timeout_i$, and $\rho_i$ retries.
(Multiple clients are averaged into one.)
Each server has a service rate $\mu_i$ and $c_i$ threads.
Server $i$ forwards the request to $i+1$, and blocks until
the downstream servers have finished processing the request.

We write $\lambda$, $\mu$, and $c$ for the $K$-dimensional vectors
of arrival rates, service rates, and threadpool sizes, respectively.


\begin{figure*}[t]
	\centering
	\includegraphics[scale=.37]{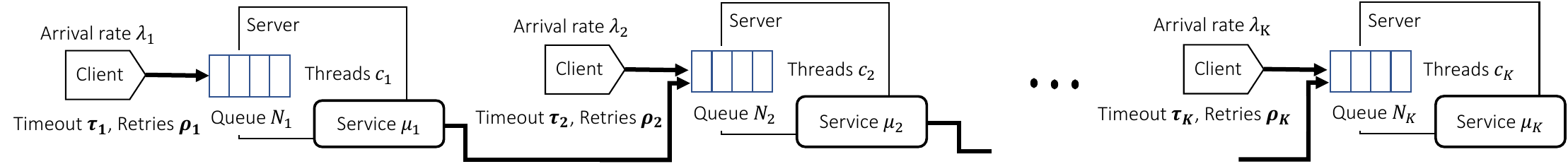}
	\caption{Pipelined servers in the DSL.} 
	\label{fig:multi-server-sketch}
\end{figure*}


The CTMC model of the program has a queue and an orbit
for each server.
We write $S_i\coloneqq \set{(u_i,v_i)\mid u_i\geq 1,  v_i\geq 1}$, where $u_i$ and $v_i$ denote the number of requests and orbit of the $i^{th}$ server, respectively.
The overall state space is $S \coloneqq \prod_{i} S_i$. We write $s_{-i}$ for the components of the state $s\in S$ without $(u_i, v_i)$.
We overload the notation and refer to the function that projects states $s\in S$ into the corresponding queue size and orbit size of each server using the notations $u_i\colon S\rightarrow \nats$, $v_i\colon S\rightarrow \nats$.
We also define the functions
$u\colon s\mapsto (u_1(s), \dots, u_K(s))$ and $v\colon s\mapsto (v_1(s), \dots, v_K(s))$.


Modeling transitions requires some thought.
The key issue is that the arrival rates of downstream
servers are affected by the service rates of upstream
servers and the service rates of upstream servers
also depend on the service rates of downstream servers.
Thus, we define \emph{effective} service and arrival
rates that summarize the dependencies.
Since the graph is acyclic, we can compute
the effective rates by a linear pass.

The effective processing rate for the $i^{th}$ server is defined as the minimum between the service rate of a server and the effective service rate of its downstream server:
\begin{equation}\label{eq:closed_multi_server_effective_mu}
    \bar\mu_i(u) \coloneqq\begin{cases} 
    \min(\min(c_i,u_i)\times\mu_i,\bar\mu_{i+1}(u)) & \text{$i<K$}\\ 
    \min(c_i,u_i)\times\mu_i & i=K.
    \end{cases}
\end{equation}
%
Similarly, the effective arrival rate at the $i^{th}$ server is defined as
\begin{equation}
\label{eq:lambda_def_multi_server}
	\bar\lambda_i(u) \coloneqq
	\begin{cases}
    \lambda_i & i=1\\
	\lambda_i + \min(\bar\lambda_{i-1}(u), \bar\mu_{i-1}(u)) & i>1.
    \end{cases}
\end{equation}

In order to compute the failure probability, let $\ell_i(u)$ denote the latency corresponding to the $i^{th}$ server. Failure probability is defined as
\begin{equation}\label{eq:failure_def_multi_server}
	r_i(u) \coloneqq\prob(\ell_i(u)>\timeout_{i}).
\end{equation}
We can use Chebyshev's inequality to over-approximate the value of $\prob(\ell_i(u)>\timeout_{i})$. Since request processing across servers is independent, we can define the mean and variance of the overall processing time for a request in the $i^{th}$ server as follows:
\begin{align}   
\mathsf{MT}_{i}(u) \coloneqq \sum_{l\geq i}u_l/\bar\mu_l \quad\quad    \mathsf{Var}_{i}(u) \coloneqq {\sum_{l\geq i}(u_l/\bar\mu_l)^2}.
\label{eq:open_mean_std}
\end{align}
Now, using Chebyshev's inequality we have the following:
\begin{equation}\label{eq:chebyshev_ineq}
\prob(\ell_{i}(u)\geq\timeout_{i})\leq1/\zeta_i^2,
\end{equation}
where 
$$\zeta_i(u)=\max\left\{1, \left(\timeout_{i}-\mathsf{MT}_{i}(u)\right)/\sqrt{\mathsf{Var}_{i}(u)}\right\}
$$ 
We define $\bar r_i(u)\coloneqq 1/\zeta_i^2$ as an upper bound over the failure probability.

Now, we are able to characterize the generator matrix $Q$ as follows:
\begin{equation}\label{eq:multi_server_Q}
	Q(s)(s')=\begin{cases}
		\bar\lambda_i(u(s)) \bar r_i(u(s))&u_i(s')=u_i(s)+1, v_i(s')=v_i(s)+1, s_{-i}=s'_{-i}\\
		\bar \lambda_i(u(s)) (1-\bar r_i(u(s)))&u_i(s')=u_i(s)+1, v_i(s')=v_i(s), s_{-i}=s'_{-i}\\
		\bar\mu_i&u_i(s')=u_i(s)-1, v_i(s')=v_i(s), s_{-i}=s'_{-i}\\
		\alpha_i v_i(s) \bar r_i(u(s))&u_i(s')=u_i(s)+1, v_i(s')=v_i(s), s_{-i}=s'_{-i}\\
		\alpha_i v_i(s) (1-\bar r_i(u(s)))&u_i(s')=u_i(s)+1, v_i(s')=v_i(s)-1, s_{-i}=s'_{-i}\\
		(1-\alpha_i) v_i(s)&u_i(s')=u_i(s), v_i(s')=v_i(s)-1, s_{-i}=s'_{-i}\\
        0&\text{otherwise,}
	\end{cases}
\end{equation}
where $1 \leq i \leq K$. 
We deal with bounded queues as before by disabling transitions that go above the bounds.

\begin{remark}[Finite state CTMCs]
An important observation is that our CTMC model
has a finite state space, so that we can use algorithmic
techniques for finite state CTMCs to analyze programs in
our DSL.
While the queue size is bounded because servers come with
natural bounds on the number of jobs in the queue, 
the orbit size can, in principle, grow without bound
and the CTMC may be transient (diverge to larger and larger states).

However, we have proven that the CTMC model with an unbounded orbit is positive recurrent and ergodic, 
and hence the
existence of stationary distribution is guaranteed. 
This means that every state is visited almost surely, we cannot ``get stuck forever'' in some mode---a full queue will drain almost surely.
Thus, we are justified in studying the behavior of the finite-state model that imposes an upper bound on the orbit.
\end{remark}

\section{Data-Driven Calibration of the CTMC}
\label{sec:ctmc_calibration}

Unlike usual programming models, the semantics of the CTMC constructed
from a program does not coincide with the simulation semantics.
This is unavoidable, due to the modeling decisions to 
abstract away simulator state to achieve Markovian dynamics
and tractable algorithmic analysis.
However, we would still like some empirical correspondence
between the model and the simulator, so that predictions
from the model are meaningful.


In this section, we present a method for \emph{calibrating} the CTMC, using a finite set of trajectories generated by discrete-event simulation (DES) of a program. 
Our modeling and calibration uses the structure of the CTMC
as a prior, but learns parameters of the model that minimize
the deviation from simulation data.
A key advantage of our approach is that it yields a continuous-time model, even when the available data consists solely of non-timed observations--i.e., sequences sampled at fixed intervals.

Let $\theta$ be the vector of
real-valued constants appearing in a program $P$.
Let $\Theta\subset \reals^{|\theta|}$ be a compact feasibility set from which the constants $\theta$ may be chosen. 
We write $\M(P^\theta)=(S, Q^\theta)$ to denote the CTMC
defined in \cref{sec:ctmc_models} for $P$.
We denote by $X^\theta(t)$ for the (random) state of $\M(P^\theta)$ at time $t\in \reals_{\geq 0}$.


To calibrate the CTMC, we choose a set of initial states $\set{s_0^{(1)}, s_0^{(2)},\dots,s_0^{(Z)}}$. 
For each $1\leq i \leq Z$, we run the
simulator $\simulator(P^\theta)$ $M$ times, to produce $M$ 
simulated trajectories, each of length $L\in \nats$ and sampled regularly with respect to a chosen sampling time $\tick>0$. 
Note that $Z$, $M$, $L$, and $\tick$ are hyperparameters for the calibration.

The simulator state contains detailed information, e.g., the actual sequence of requests in a queue, their id's, and so on.
We instrument the simulator state to capture 
the number of requests in each queue and the number of retries occurring in the system, to match the CTMC state.
While the queue size is ``exact,'' the number of retries
in the system is an approximation to the CTMC's notion of orbit size.
(For one, a retry happens in the simulator after a timeout,
but the CTMC can add an element to the orbit upon arrival.)
We have seen that this difference between the calibrated CTMC and the simulator is negligible.

For every $1\leq i \leq Z$ and $1\leq j \leq M$, we
write $\hat{X}^{\theta}_{i,j}(k\tick)$ for the abstracted
simulator state (only the number of jobs in the queue and the number of retries) at time steps $0\leq k \leq L-1$.
Note that $\hat{X}^{\theta}_{i,j}(0)$ is the abstraction of
$s_0^{(i)}$ for every $1\leq j \leq M$. 

Next, for every $1\leq i \leq Z$ and $0\leq k \leq L-1$, we  
compute the empirical average 
\[
y_i^\theta(k\tick)\coloneqq \frac{1}{M}\sum_{j=1}^M \hat{X}^{\theta}_{i,j}(k\tick)
\]
This gives the averaged dynamics of the simulator over $M$ runs.

We would like to ``match'' this average simulator dynamics to average CTMC states at corresponding times.
The corresponding CTMC states are computed as 
\[
y^\theta_i(k\tick)\coloneqq\mathbb{E}(X^\theta(k\tick)\mid X^\theta(0)=s_0^{(i)}),
\]
where the expectation is computed using the
matrix exponential of the generator matrix of the CTMC.

For a program $P^{\theta_0}$, our aim is to find $\theta^\ast\in \Theta$ such that (1) $\theta^\ast$ is close to $\theta_0$, and (2) the average output trajectories of $\M(P^{\theta^\ast})$, i.e., $y^{\theta^\ast}_i(k\tick)|_{k=0}^{L-1}$, match as closely as possible with the trajectories of $\simulator(P^{\theta_0})$, i.e., $\hat{y}^{\theta_0}_i(k\tick)|_{k=0}^{L-1}$, for every $1\leq i \leq Z$. 

Formally, we solve the following optimization problem that minimizes the loss:
\begin{align}\label{eq:optimization_ctmc_cal}
	&\min_{\theta\in \Theta}\gamma_1\|\theta-\theta_0\|_2^2 + \gamma_2\sum_{i=1}^Z\sum_{k=0}^{L-1} \|y_i^\theta(k\tick)- \hat y_i^{\theta_0}(k\tick)\|_2^2,  
\end{align}
where $\gamma_1, \gamma_2 \in \reals_{> 0}$ denote the relative importance of the first and second terms in the objective function above.
It is worth noting that our calibration method produces a continuous-time model despite not requiring holding time information in the data trajectories.

        \section{Algorithmic Analysis for Metastability}
\label{sec:analysis}

Sections~\ref{sec:ctmc_models} and~\ref{sec:ctmc_calibration} give us a way to abstract request-response systems into a model that is amenable to
formal algorithmic analysis.

In this section, we present both qualitative and quantitative analyses to examine whether the 
calibrated CTMC model exhibits metastable behavior
and to predict properties such as recovery time.
Along the way, we provide a formal definition of 
metastability.

\subsection{Qualitative Analysis through Visualization}

In the context of dynamical systems theory, visual flow analysis is a powerful tool for identifying qualitative behaviors such as stability. Such analysis is typically applicable to low-dimensional systems, generally of order three or less. 
%
Unfortunately, the Kolmogorov equations
defining the dynamics are over a very high-dimensional state space (the number of states of the CTMC), and we cannot directly visualize this 
dynamics.
In what follows, we introduce an efficient approach for performing flow analysis on CTMC models arising from our DSL.

Let us focus on a single server.
The state space is two dimensional and can be interpreted as a two-dimensional grid, with one dimension corresponding to the queue and the other to the orbit. Furthermore, the transitions are \emph{sparse}: a state $(u,v)$ can only reach
its neighbors that differ by at most one in a coordinate.
This suggests a visualization of the \emph{aggregate} dynamics in a two-dimensional plane as follows.

For an arbitrary state $(u,v)\in S$, let us define 
\begin{equation}\label{eq:fq_fo_def}
\begin{bmatrix}
    f_q(u,v)\\f_o(u,v)
\end{bmatrix} \coloneqq \sum_{(u',v')\neq (u,v)}Q((u,v), (u',v'))\begin{bmatrix}
        u'-u\\v'-v
    \end{bmatrix}.
\end{equation}
The two components capture the dynamics in the ``queue dimension'' and the ``orbit dimension,'' respectively.
We now define 
\begin{equation}\label{eq:mag_ang_def}
    \mathcal{A}(u,v) \coloneqq  \sqrt{f_q^2(u,v) + f_o^2(u,v)},\quad \theta(u,v) \coloneqq \arctan(f_o(u,v)/ f_q(u,v)).
\end{equation}
We visualize the dynamics of a CTMC by plotting,
at any selected $(u,v)$, an arrow whose magnitude 
corresponds to $\mathcal{A}(u,v)$ 
and whose orientation corresponds to the angle $\theta(u,v)$.
Since we normalize the magnitude, we also use a color scheme to visually present the magnitude.

The complexity of the procedure depends on the sampling density, but finding the visualization at a single point is
independent of the size of the CTMC. Thus, the visualizations are produced in a matter of milliseconds for large (100's of thousands of states) CTMCs.

For multi-server systems with $K>1$ servers, there are $K$ two-dimensional components.
In this case, we visualize the flows for one server at a time, fixing the state components corresponding to the other servers to fixed values.

The visualization highlights only the dominant flows in a deterministic manner and may obscure the fact that the underlying dynamics are inherently stochastic. Therefore, we complement the visualization with analytical tools, as described in the following subsection.

\subsection{Quantitative Analysis I: Expected Hitting Times}

For a CTMC, we define the \emph{hitting time} for a set $D$ as the first (nonzero) time when the chain visits $D$, starting at some state $X(0)=x$:
\begin{equation}
 \label{eq:hitting_time}
  \tau^x_D = \inf \set{ t > 0 \mid X(t) \in D, X(0) = x}. 
 \end{equation}
One can formulate the computation of
recovery times as \emph{expected hitting times}
in the CTMC from a state corresponding to the full queue and high orbit to a state where the queue is empty or corresponds to the average queue size in the stationary distribution.
The expected hitting time can be calculated
by solving a system of linear equations \cite{Norris}.

However, simply calculating expected hitting times does not capture metastability.
First, expected hitting times increase with the queue and orbit sizes. 
Second, the expected hitting time increases
exponentially for unstable systems.
Thus, just because the hitting
times increase does not mean the system has metastable states.

Instead, we can capture the different time scales by considering
the \emph{relative} expected hitting times.
We can call a CTMC metastable w.r.t.\ a set $D \subset S$ of states if
\begin{equation}\label{eq:metastability_def_ht}
|S|\frac{\sup_{x\not\in D} \mathbb{E} [\tau^x_D]}{\inf_{x\in D} \mathbb{E} [\tau^x_{D\setminus \set{x}}]} << 1,
\end{equation} 
that is, if the expected time scale of traveling between different states in $D$ is much larger than reaching some state in $D$ from outside of $D$.

While one can use this definition, it is not ideal due to the fact that solving the linear equation for expected hitting times becomes numerically unstable around metastable paramterizations.
Since one of our aims is to \emph{predict} hitting times, 
we would like an alternate characterization that allows us
to approximate the predictions in a more numerically stable way.

\subsection{Quantitative Analysis II: Metastability and Eigenvalues}

We now provide a characterization of metastability in CTMCs
and connect the characterization to the spectral properties of the
generator matrix.
Our definition is inspired by the analysis of metastability for \emph{discrete-time} Markov chains~\cite{BEGK2002,MetastabilityBook}, and we extend the definition to the continuous-time setting.

\paragraph{A Characterization of Metastability}
For a state $x \in S$ and a set $D \subset S$, we define the \emph{escape probability} from $x$ to $D$ as 
$\prob(\tau^x_D < \tau^x_x)$.
That is, the escape probability is the probability that, if the
chain starts at $x$, it visits $D$ before it visits $x$ again.

\begin{definition}[Metastability] 
A finite CTMC $\M = (S, Q)$ is \emph{$\rho$-metastable with respect to a set $D \subset S$} if
	\begin{equation}
		\label{eq:meta_stability}
		|S| \frac{\sup _{x \in D} \mathbb{P}\left(\tau^x_{D \backslash \set{x}}<\tau^x_x\right)}{\inf _{y \notin D} \mathbb{P}\left(\tau^y_D<\tau^y_y\right)} \leq \rho \ll 1.
	\end{equation}
\end{definition}

Intuitively, a CTMC is $\rho$-metastable w.r.t.\ $D$ 
if any state in $D$ will visit a different state in $D$ at a time scale that is much larger than the time scale for it to visit itself
or the time scale for any state outside of $D$ to visit some state in $D$.
That is, each state in $D$ acts as an ``attractor'': 
once the state is at $x\in D$, it is more likely that $x$ is revisited
before some other state $y\in D$ is visited (although such visits happen probability one).
Moreover, any state outside of $D$ is attracted to some state in $D$,
again at a time scale faster than a visit between two different states in $D$.

We note that our characterization simply determines the metastable state, without stating whether a state is ``good'' or ``bad''.
In our application, metastable states typically correspond
to ``full queue'' and ``average queue'' states.
Since system performance is bad in the ``full queue'' metastable
states, we consider these states undesired and call them metastable \emph{failures}.

\paragraph{Relating Metastability to Eigenvalues}
Next, we show that our notion of metastability can be predicted by looking at the eigenvalue structure of the generator matrix.
The following theorem is an extension of \cite[Theorem 8.43]{MetastabilityBook} to the case of CTMCs, where we also
need the methods of \cite[Theorem 1.2]{eckhoff2002low}.

\begin{theorem}
\label{thm:metastability-spectral}
Let $\M = (S, Q)$ be an ergodic, finite CTMC.
Let $D$ be a set of metastable points, $|D| = k$.
Define $D_k = D$, and 
\[
D_{\ell-1}:=D_{\ell} \backslash \set{x_{\ell}},\quad x_{\ell}:=\operatorname{argmax}_x\left[\mathbb{P}_x\left(\tau_{D_\ell \backslash x}<\tau_x\right), x \in D_{\ell}\right],\quad \forall\ell\in\{2,3,\ldots,k\}
\]
Then $-Q$ has $k$ eigenvalues 
$0 = \lambda_1 < \lambda_2 < \ldots < \lambda_k$,
and
\begin{equation}
\lambda_{\ell}=
\frac{1}{\mathbb{E}\left[\tau_{D^{x_{\ell}}_{\ell-1}}\right]}[1+O(\frac{\rho}{|S|})], \quad \ell \in \set{2, \ldots, k}.
\end{equation}
\end{theorem}

\Cref{thm:metastability-spectral} states that $k$ metastable states are characterized by a cluster of $k$ eigenvalues near $0$.
As an example, consider the CTMCs from \cref{sec:overview}:
the metastable version with arrival rate 9.5 RPS and the throttled
version with arrival rate 8 RPS.
\cref{fig:spectral_gap_motiv_ex} shows the two dominant 
eigenvalues of the two CTMCs.
The largest eigenvalue is $0$ for both CTMCs, since \emph{every} CTMC has 0 as the dominant eigenvalue.
However, for the first CTMC, the second largest eigenvalue is close to 0, whereas for the second CTMC, it is away from 0.

\begin{wrapfigure}{R}{6.5cm}
\vspace*{-1em}
\includegraphics[scale=.3]{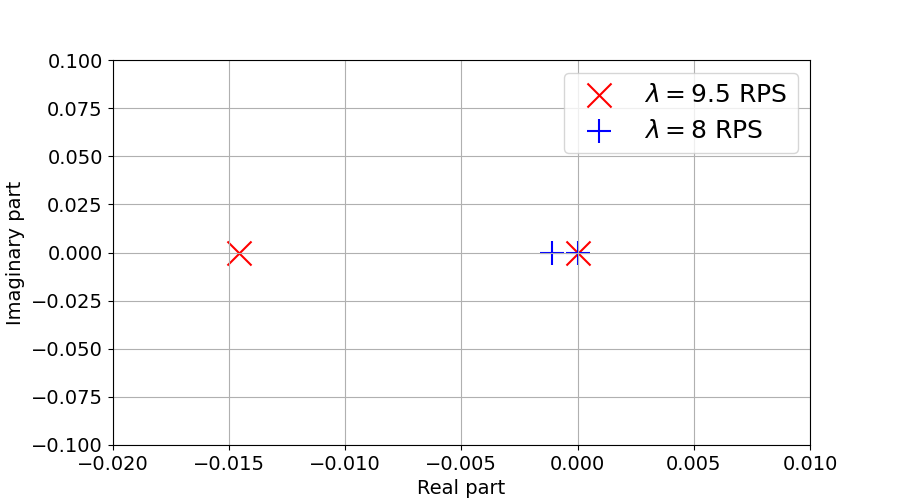}
\caption{Comparison of the two dominant eigenvalues of two CTMCs
from \cref{sec:overview}.
}
\label{fig:spectral_gap_motiv_ex}
\end{wrapfigure}

A second benefit of the spectral characterization is that we can estimate the \emph{mixing time} for a CTMC using its eigenvalues.
Informally, the mixing time measures how long the chain takes to reach its stationary distribution
and can be used as an approximation for hitting times, since
the latter can be numerically unstable.
Formally, for every $\epsilon > 0$, we can show
\[
\hspace*{-6cm}
\tau_{\mathrm{mix}}(\epsilon) \geq \log(1/2\epsilon)/|\mathrm{Re}(\lambda_{2})|
\]
where the mixing time $\tau_{\mathrm{mix}}(\epsilon)$ gives the time 
taken to reach within $\epsilon$ total variation distance of the stationary distribution and $\lambda_2$ is the nonzero eigenvalue with smallest real part.

\paragraph{A Remark About Implementation}
    Computing both expected hitting times and mixing times requires tools from linear algebra---such as solving linear systems of equations or computing the eigenvalues of the CTMC's generator matrix. In CTMC models for server systems, transitions occur only between neighboring states, resulting in a sparse structure. Consequently, the generator matrix is sparse, with the majority of its entries equal to zero. This sparsity can be exploited to achieve significant computational speed-ups by leveraging techniques from black-box linear algebra tailored for sparse matrices.

	\section{Experimental Results}\label{sec:experiments}

We now describe our experiences in characterizing metastable configurations for different system parameters and the effect
of recovery policies on recovery time after a metastable failure.
We answer the following research questions.

\begin{enumerate}
\item[RQ1] 
Is the CTMC model in Section~\ref{sec:ctmc_models} faithful to the behavior of discrete-event simulations? If not, can the calibration method in Section~\ref{sec:ctmc_calibration} compensate for the inaccuracies?

\item[RQ2] can we use our analysis to understand how system configurations 
affect metastable behaviors in a request-response system?
Do the quantitative estimations reinforce the qualitative visualizations?

\item[RQ3] how does metastability analysis aid in designing a \emph{recovery policy} that enables fast recovery after a temporary fault scenario?

\end{enumerate}


\subsection{CTMC Calibration 
}
\begin{figure}[!htbp]
\centering
\begin{minipage}{1\textwidth}
\centering
\includegraphics[scale = .4]{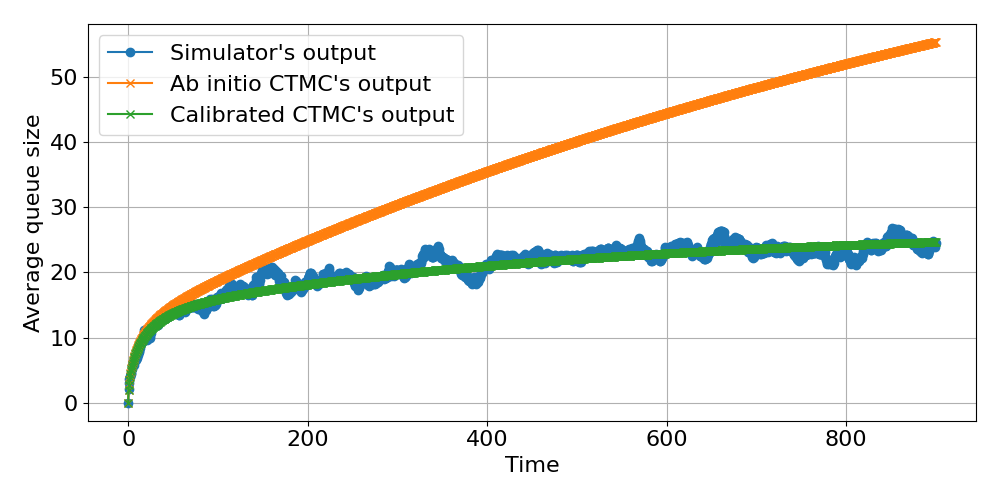}
\caption*{(a)}
\end{minipage}

\begin{minipage}{1\textwidth}
\centering
\includegraphics[scale = .4]{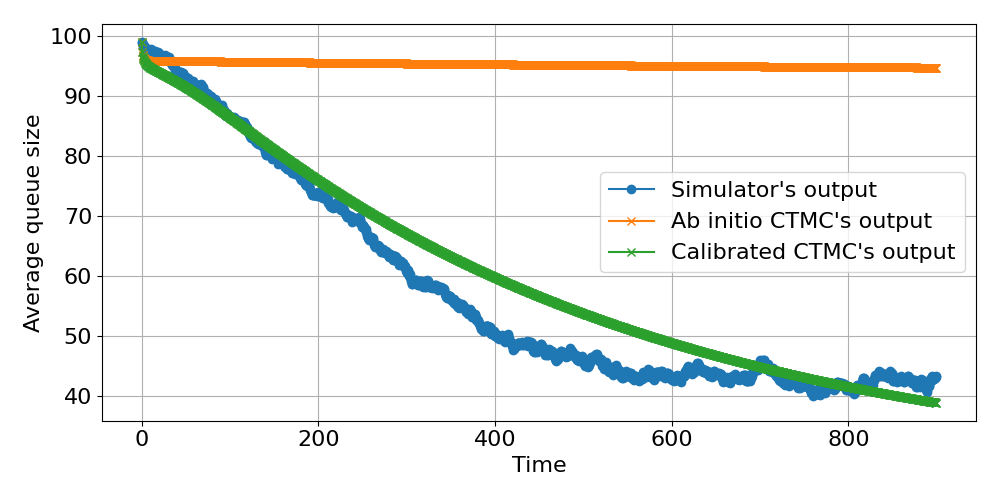}
\caption*{(b)}
\end{minipage}
\caption{Comparison of trajectories representing the average number of requests in the system, as produced by the ab initio and calibrated CTMC and the discrete-event based simulator, under two initial queue and orbit conditions: (a) empty and (b) full.
}
\label{fig:calibration_motiv_ex}
\end{figure}
In this section, we examine how the CTMC calibration method proposed in Section~\ref{sec:ctmc_calibration} influences the accuracy of the resulting CTMC model. To this end, we consider a setting similar to the motivating example: the system parameters are given by $\lambda_0 = 9.5$, $\mu_0 = 10$, $\timeout_0 = 9$, and $\retry_0 = 3$, with queue and orbit lengths set to $100$ and $20$, respectively. Our focus is on calibrating the model with respect to the arrival rate and the timeout value. Specifically, we define the nominal parameter vector as $\theta_0 = \begin{bmatrix}\lambda_0&\timeout_0\end{bmatrix} = \begin{bmatrix}9.5& 9\end{bmatrix}$. For the feasibility set, we fix $\mu$ and $\retry$ to their nominal values and define the search space as $\Theta = (9, 10) \times (7, 11)$. 

To solve the optimization problem described in Eq.~\eqref{eq:optimization_ctmc_cal}, we use covariance matrix adaptation evolution strategy (CMA-ES), which is an efficient optimization method that belongs to the class of evolutionary algorithms \cite{Hansen:2005}. It is stochastic and derivative-free, and can handle non-linear, non-convex, or discontinuous optimization problems. 

We consider two initializations, $Z = 2$, corresponding to the cases where both the queue and orbit are either empty or full. For each initialization, we generate $M = 100$ simulation runs, each with a duration of $L =1800$ and a sampling interval of $T_s = 0.5$. After 30 iterations of Running CMA-ES, 
the calibrated parameters converge to $\theta^\ast = \begin{bmatrix}\lambda^\ast& \timeout^\ast\end{bmatrix} = \begin{bmatrix}9.43& 10.54\end{bmatrix}$, with a total execution time of 978 seconds. 

\Cref{fig:calibration_motiv_ex} shows subsequent simulations of the program and the two CTMCs, before and after calibration.
While the trajectories generated by $\M^{\theta_0}$ deviate from those produced by the discrete-event based simulator, the calibrated model $\M^{\theta^\ast}$ produces trajectories that closely align with them.

We conclude that, although the output of the \emph{ab initio} CTMC model from \cref{sec:ctmc_models} does not always align with that of the discrete-event simulator, the calibration method in \cref{sec:ctmc_calibration} significantly reduces this mismatch (RQ1).

\begin{figure}[!t]
	\centering
    \begin{minipage}[]{0.48\textwidth}   \includegraphics[width=1\textwidth]{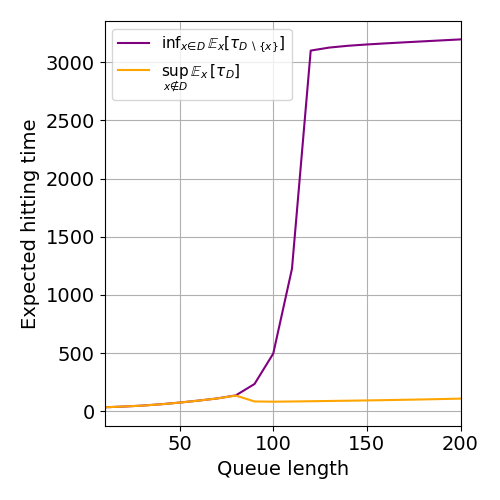}
    \caption*{(a)}
    \end{minipage}
    \begin{minipage}[]{0.48\textwidth}   \includegraphics[width=1\textwidth]{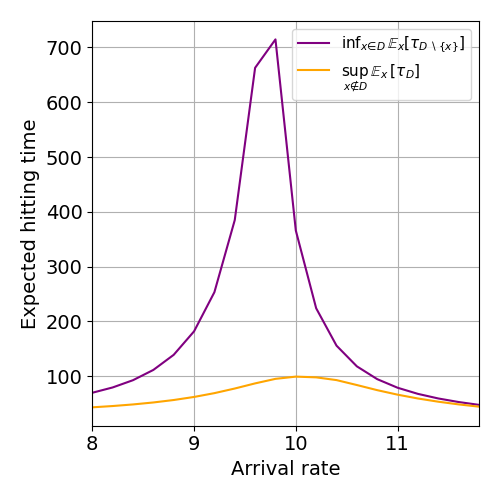}
    \caption*{(b)}
    \end{minipage}\\
    \begin{minipage}[]{0.48\textwidth}   \includegraphics[width=1\textwidth]{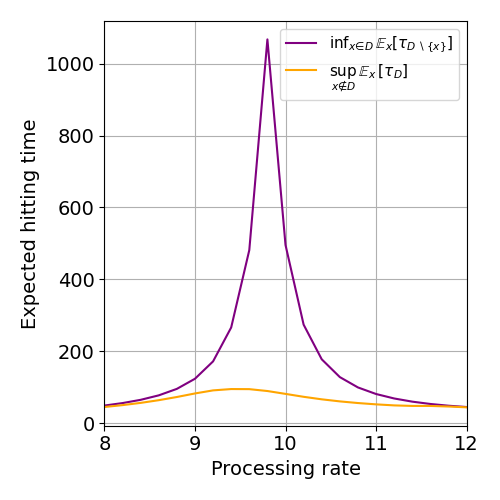}
    \caption*{(c)}
    \end{minipage}
    \begin{minipage}[]{0.48\textwidth}   \includegraphics[width=1\textwidth]{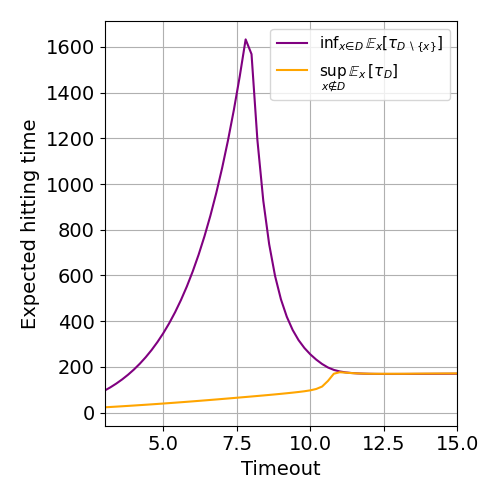}
    \caption*{(d)}
    \end{minipage}
	\caption{Illustration of variations in expected hitting times used to verify the system's metastability with respect to the set \( D = \{\low, \high\} \), where \( \low \) and \( \high \) represent the CTMC states corresponding to, respectively, the empty and full queue and orbit. The variations are shown with respect to (a) queue length, (b) arrival rate, (c) processing rate, and (d) timeout. 
    }
	\label{fig:parametrization_ht_s_u}
\end{figure}
\subsection{Effect of System Parameters on Metastability}\label{subsec:parametrization_analysis}

First, we evaluate how metastability emerges as system parameters change on the running example.
We focus on this example, but our analysis techniques and run times are similar for other models.

We set the parameters of the model to the following \emph{nominal} values, as in Sections~\ref{sec:intro} and \ref{sec:overview}: 
the arrival rate $\lambda_0=9.5$ RPS, 
processing rate $\mu_0=10$ RPS, 
maximum number of retries is $3$, 
timeout $\timeout_0 = 9$s, 
queue length $100$,
and orbit length $20$. 
We set \( D = \{ \low, \high \} \), where \( \low \) corresponds to the CTMC state in which both the queue and orbit are empty, and \( \high \) corresponds to the state where both the queue and orbit are full. To measure the effect of parametrization on the system's metastability index, \cref{fig:parametrization_ht_s_u} illustrates how the numerator and denominator in \cref{eq:metastability_def_ht} vary with respect to queue length, arrival rate, processing rate, and timeout values. Intuitively, a parameterization corresponds to metastability if $\sup_{x \notin D} \mathbb{E}_x [\tau_D]$ is small, indicating that states in $S \setminus D$ quickly reach $D$, and $\inf_{x \in D} \mathbb{E}_x [\tau_{D \setminus \{x\}}]$ is large, indicating that the expected travel time between $\low$ and $\high$ (in both directions) is small. 

\cref{fig:parametrization_ht_s_u} shows that reducing the processing rate and timeouts, or increasing the arrival rate and queue length, leads to the emergence of metastable behaviors. 
These results are in accordance with visualizations.
Notice that increasing the queue length beyond a certain point causes the system to remain metastable as the queue length increases, in contrast to the effect of the other parameters. 
Specifically, \cref{fig:parametrization_ht_s_u}(a) shows that for queue lengths greater than 90, the system remains metastable,
as expected from the visualization.
In comparison, for other parameters, from \cref{fig:parametrization_ht_s_u}(b-d), metastability occurs when $\lambda \in (9, 10.5)$, $\mu \in (9.5, 10.5)$, and $\timeout \in (5, 10)$. This is an important observation: increasing the queue length, while increases 
the expected hitting time between $\low$ and $\high$, does not make \emph{either} of them a universal attractor for the entire state space. In contrast, changing the other parameters, i.e., arrival and processing rates and timeout, alters the relative attraction between the two, making one of $\low$ and $\high$ the universal attractor (based on stability/instability).
This is key point: metastability is different from stable ($\low$ is the only attractor) and unstable ($\high$ is the only attractor) behaviors!

We conclude that our formal notion of metastability captures observed metastable behaviors in systems and the CTMC
model helps us navigate the space of parameters (RQ2).

\subsection{Recovery}\label{subsec:recovery}

Next, we study the effect of metastable modes on the \emph{recovery time} of the system.
Specifically, we consider the expected time reach a queue size
less than 10\% of the maximum queue length
when starting from the state $\high$. 
In principle, configurations that lead to a large recovery time should be avoided. 
In practice,  
one may prefer to allow such configurations to optimize performance, 
while accelerating recovery from $\high$ by selecting an appropriate \emph{recovery policy}.
%
A recovery policy adjusts the rates with the goal of returning the system to $\low$. A default recovery policy does
nothing; but a recovery policy can throttle the arrival rate or increase the processing rate.


Figure~\ref{fig:return_all} shows the effect of changing the arrival rate $1\leq \lambda_r \leq 10$, the processing rate $8\leq \mu_r \leq 12$, and the timeout $5\leq \timeout_r \leq 15$ when the queue is full. 
We plot the expected hitting time and include a range around it that has the standard deviation of the hitting time as its radius. 
It can be observed that the recovery time increases with a lower processing rate, shorter timeouts, and a higher arrival rate. 

\begin{figure}[t]
    \centering
    \begin{minipage}{0.32\textwidth}
        \centering
        \includegraphics[width=\textwidth]{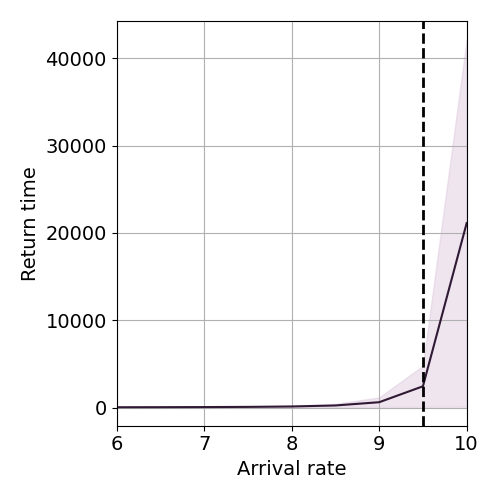} 
        \label{fig:return_a}
    \end{minipage}
    \begin{minipage}{0.32\textwidth}
        \centering
        \includegraphics[width=\textwidth]{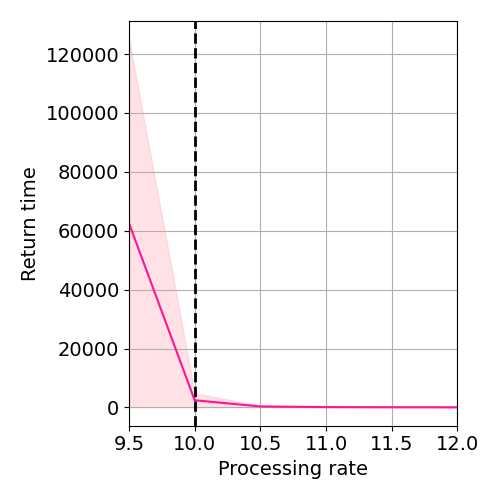} 
        \label{fig:return_b}
    \end{minipage}
    \begin{minipage}{0.32\textwidth}
        \centering
        \includegraphics[width=\textwidth]{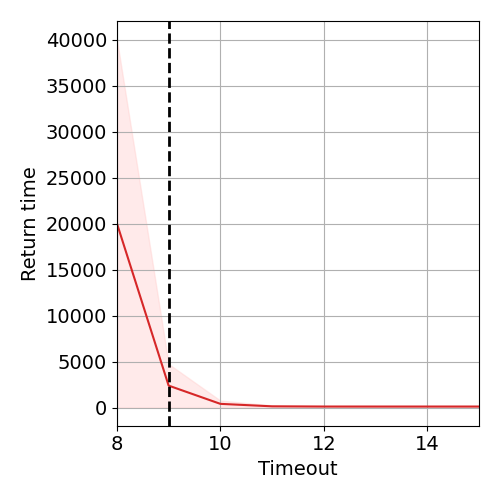} 
        \label{fig:return_c}
    \end{minipage}
    \caption{Effect of setting parameters during recovery on the system recovery time. 
    }    \label{fig:return_all}
\end{figure}

\begin{figure}[!htbp]
	\centering
         \includegraphics[scale=0.4]{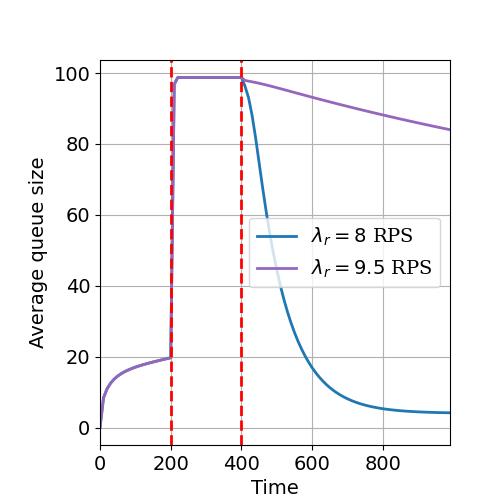}
	\caption{Two recovery policies. The red lines indicate a load spike. Recovery starts at time 400s with default rate 9.5 RPS and throttled rate 8 RPS.}
	\label{fig:good_vs_bad_policy}
\end{figure}
%

Figure~\ref{fig:good_vs_bad_policy} shows the effect of two concrete recovery policies, generated from the CTMC
by solving the Kolmogorov equations.
The default policy has a long recovery time since $\lambda_r=9.5$ corresponds to a metastable configuration.
Throttling the arrival rate to $\lambda_r=8$, which corresponds to a stable configuration, 
causes rapid recovery.

We conclude that the CTMC-based exploration helps us analyze the effect of recovery policies by predicting average recovery times (RQ3).

\begin{figure*}[t]
	\centering
	\includegraphics[scale=.37]{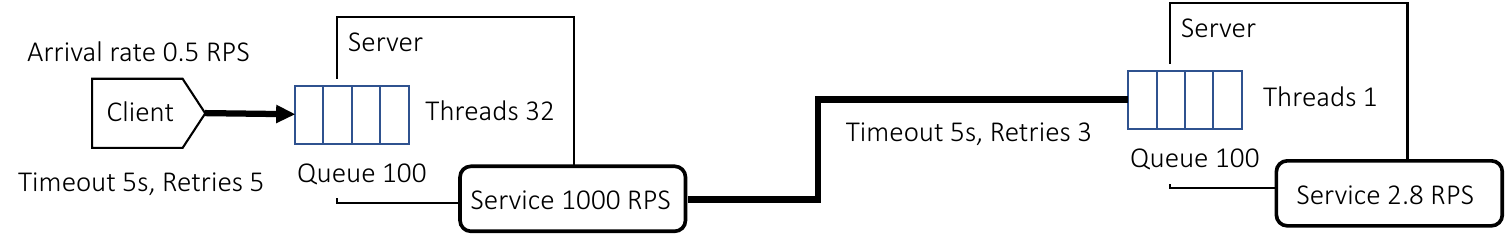}
	\caption{Multi-server system considered in \cref{subsec:experiments_multi_server}. 
    }
	\label{fig:two-server-sketch}
\end{figure*}
\subsection{Metastability in Multi-Server Systems}\label{subsec:experiments_multi_server}

To show that our analysis scales to more complex systems, we now move to a multi-server example.  
We consider an example, inspired by an industrial service, with two servers connected serially (\cref{fig:two-server-sketch}). 
The first server has $32$ threads and receives requests at rate $0.5$ RPS.
Each thread, after some quick processing (rate $1000$ RPS), forwards requests to the second server and 
waits until the second server is done.
The second server has processing rate $2.8$ RPS. 
We set the queue and orbit lengths, respectively, to be $100$ and $20$ for both servers. 
We set timeout to be $5$s for both servers, and 
the maximum number of retries to $5$ and  $3$, respectively.
The effective service rate of the first server is determined by the processing rate of the second.

\cref{fig:viz_two_server} presents a visualization of the stochastic dynamics over the state space of the second server, assuming that the first server is in the $\high$ state. 
The visualization is generated in milliseconds.
Since all 32 threads of the first server are in use, the second server's queue contains at least 32 pending requests. As a result, the range of queue lengths in the visualization only includes values greater than 32. At first glance, the system might appear stable. However, a closer inspection reveals that for states where the orbit length is near 15, the transitions between different states \emph{almost} balance each other (as indicated by the amplitude of the arrows, visible through the colorbar). This suggests that starting from states near this region, the system may become stuck for a relatively long time. To gain a more precise understanding of the system's metastability, we perform further quantitative analysis.

\begin{figure}
	\centering
         \hspace*{1cm}\includegraphics[scale=0.5]{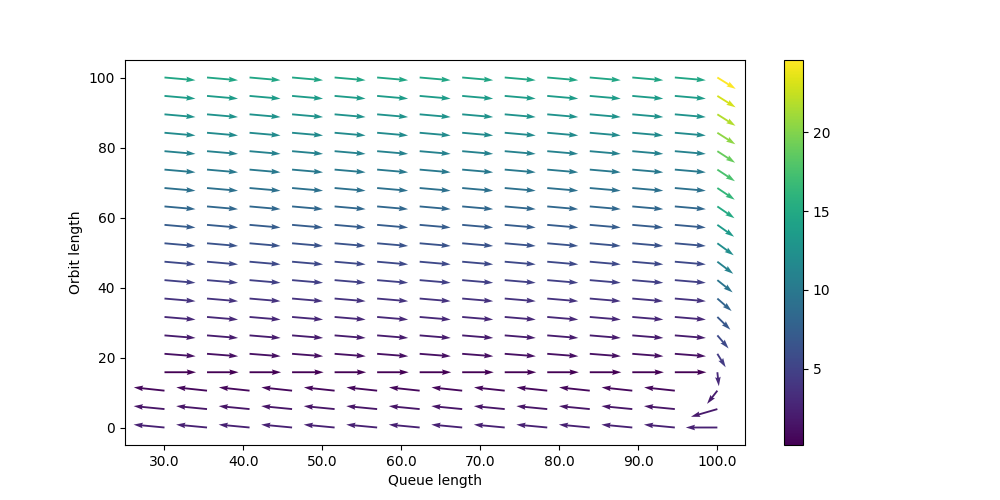}
	\caption{Visualization of the stochastic dynamics for the second server.}
	\label{fig:viz_two_server}
\end{figure}
The CTMC representing the server system corresponds to a generator matrix with  $16 \times 10^{12}$ entries, 
which is too large to keep explicitly. 
In our implementation, we use \emph{black-box linear algebra} techniques \cite{Wiedemann86} 
to perform the required computations. 
We use the system's mixing time as a proxy for detecting metastability. 
For the given parameterization, we found the mixing time to be $10^7$s, which is much larger than the time scale of the CTMC,
showing metastability.

Figure~\ref{fig:mixing-time-analysis} shows how the mixing time 
varies with queue length, orbit length and processing rate. 
Increasing the queue and orbit lengths increases the mixing time, which is expected, 
as a larger number of requests in the main queue and orbit space effectively prolongs the transition from $\high$ to $\low$. 
As in the single-server experiment, changing the processing rate does not affect the mixing time in a monotonic manner. 
Both very low and very high processing rates result in stable or unstable behaviors with 
short mixing times, while intermediate rates lead to metastable configurations.

\begin{figure}[t]
\begin{center}
	\begin{minipage}{0.32\linewidth}
		\centering
		\includegraphics[width=\textwidth]{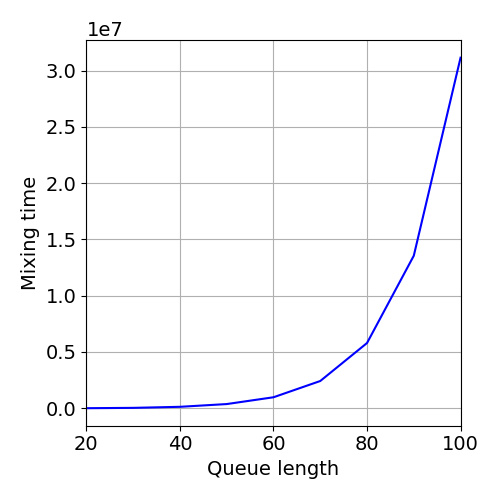}
	\end{minipage}~
	\begin{minipage}{0.32\linewidth}
	 	\centering
	 	\includegraphics[width=\textwidth]{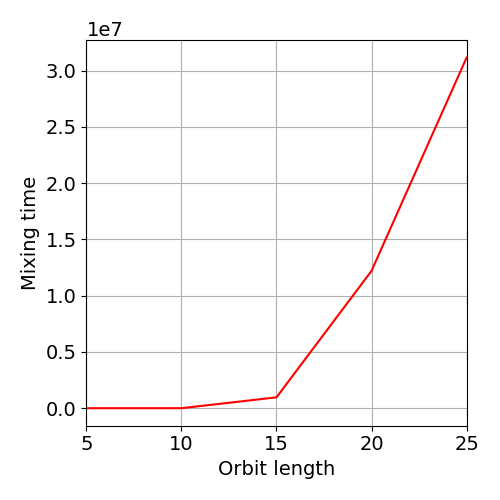} 
	 \end{minipage}~
    \begin{minipage}{0.32\linewidth}
		\centering
		\includegraphics[width=\textwidth]{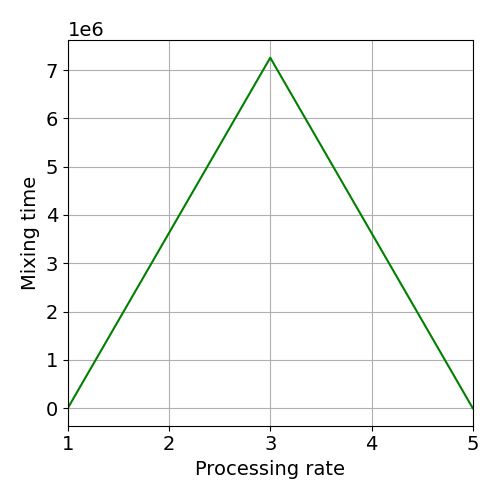} 
	\end{minipage}

 %
 \end{center}	\caption{Mixing time vs.\ (a) queue length, and (b) orbit length, (c) 
 processing rate.
	}\label{fig:mixing-time-analysis}
\end{figure}

	\section{Related Work}\label{sec:related_work}

Our work is part of a larger project that aims to understand, predict, and mitigate metastable failures in large-scale cloud infrastructure \cite{analyzing-metastable-faiures-hotos-2025}.
While we focus here on the formal modeling and analysis, the larger context
also involves tuning the simulator with a service emulator as well as connecting the analysis workflow with workload testing of the actual service.

\paragraph{Metastability in the sciences}
Metastability is a widespread phenomenon in physical systems \cite{StochasticThermodynamics,MolecularDrivingForces,FW}.
For instance, in statistical mechanics, phase transition phenomena such as magnetic hysteresis or condensation of over-saturated water vapor
are examples in which a system remains ``persistently'' in one state and then rapidly transitions into another in the presence of some rare event.

Metastability has been studied formally in the context of perturbed dynamical systems \cite{FW}.
In the context of Markov processes, Bovier et al.~\cite{BEGK1,BEGK2002,MetastabilityBook} defined metastability for discrete-time Markov chains and characterized metastability using
potential theory and spectral methods.
Spectral techniques for metastability go back to the work of Davies \cite{Davies1982a,Davies1982b,Davies1983}. 
Most of their results hold for reversible Markov chains, with more technical extensions to the non-reversible case.
Betz and Le Roux studied metastability for perturbed Markov chains, considering the asymptotics of metastability as the perturbation parameter goes to zero \cite{Betz2016}.

\paragraph{Metastable Failures in Systems}
In the context of computer systems, Bronson et al.~\cite{bronson2021metastable} introduced the term metastable failures, and gave examples and informal definitions of such failures.
They pointed out that such failures have been studied in several settings in systems and networking, often with different names, such as persistent congestion \cite{AWS2021_persistent},
retry storms \cite{Azure2021_storm}, or cascading failures \cite{cascade2016}.
Huang et al.~\cite{Huang2022systemmetastability} performed an extensive empirical study that demonstrated metastable failures are a common cause of published 
outages in many large software organizations.
They refined the informal characterizations of Bronson et al.~and reproduced such failures empirically.

\paragraph{CTMC Models for Metastability}
CTMCs have been recently proposed as models of metastable server systems \cite{habibi2023msfmodel}, and their work is close to ours and an inspiration for us.
We improve upon their work in several ways.
We generalize their CTMC model to more features, including mixtures of APIs and handling multi-server
systems, which is important in modeling real scenarios.
Furthermore, we calibrate the model using simulation.
As we point out, calibrating the model is crucial in providing quantitative predictions that match the ``ground truth'': the non-calibrated models such as theirs deviate from reality already for simple systems.

\citet{habibi2023msfmodel} also provide a definition of metastability
as the average distance from the origin at the stationary distribution.
This notion captures the distinction between stability and instability, but not metastability. 
As we remark in \cref{sec:experiments}, an unstable system has an attractor
in a $\mathsf{High}$ state, but this does not necessarily indicate metastable failures.
In contrast, we provide a definition of metastability that is mathematically robust and captures dynamics at different time scales. 

Finally, our analyses are substantially different from theirs: \cite{habibi2023msfmodel} performs Monte Carlo simulations to collect finite-horizon empirical probability distributions; in contrast, we provide both qualitative visualizations and formal analysis based on expected recovery times.

\paragraph{Queueing Systems}
The analysis of servers, requests, and queueing is the domain of \emph{queueing theory} (see, e.g., \cite{Kleinrock1975}).
Retrial queueing models \cite{Artalejo2008} capture the behavior of retry policies by maintaining a (possibly infinite) \emph{orbit} for requests waiting to be retried.
The semantics of queueing models is also given as continuous-time Markov processes.
Our point of departure from classical queueing theory is twofold.
First, we define metastability in queueing models and consider algorithms for analyzing metastability.
Classical queueing models such as M/M/c queues do not demonstrate metastable behaviors: such behaviors are seen only when we add retrials into our model.
Second, our queueing model captures features that are specific to the domain of software systems, such as many instances of the same request existing in the system
(either in the queue or in the orbit) at the same time.
In many queueing models, a request only has one instantiation in the system: a full queue causes a retry, but there is no explicit handling of timeouts due to long latencies
that add additional requests to the system while the original requests still wait in the queues. 
However, this is a common pattern in software systems.

\paragraph{Probabilistic verification}
CTMCs have long been employed as models for analyzing system performance~\cite{harchol2013performance,Hillston1995,BaierHHK05}, and a substantial body of work exists on model checking CTMCs against correctness and performance properties~\cite{Haverkort2002,BaierHKH05}. Temporal logics such as Continuous Stochastic Logic (CSL) have been widely used to specify and verify properties of CTMCs~\cite{Aziz:2000,Baier:2003}, extending traditional temporal logics with real-time constraints to reason about time-bounded behaviors. While CSL provides a powerful formalism for verification, it lacks the ability to capture multi-scale dynamics or investigate the presence of metastability. Ballarini et al.~\cite{Ballarini:2009} investigate oscillatory behavior in biochemical processes using CSL and probabilistic Computation Tree Logic (pCTL). The oscillations they model are related to recurrent behavior for a subset of the Markov chain’s state space, which are fundamentally different from properties such as almost-invariance that characterize metastable dynamics.
 In summary, existing specification formalisms and model checking techniques focus primarily on transient or steady-state behaviors and have not addressed metastability or ``almost-invariant'' behaviors.

\paragraph{Learning CTMCs from data.} Previous work has addressed the problem of learning CTMC models, generally following two distinct approaches. The first directly estimates the transition rates of the CTMC from data, while the second focuses on identifying the set of parameters that best fit the data, under a parametric formulation of the transition rates. 

Within the first category, \cite{WEI2002129} proposes a method for learning continuous-time hidden Markov models aimed at performance evaluation. In their framework, time series observations are treated as periodic samples taken at fixed intervals. The learning procedure proceeds in two stages: first, a maximum likelihood estimation algorithm is used to infer the transition probability matrix of a discrete-time hidden Markov model; then, the generator matrix of the CTMC is derived from the learned transition matrix. In contrast, \cite{Bacci2023} introduces a more direct approach that simplifies the process by learning the CTMC generator matrix without transitioning through a discrete-time model. Related to the second category, the Evolving Process Algebra \cite{Marco:2011} framework uses genetic algorithms to find parametrization of models written in the PEPA language \cite{Hillston:1996}, such that the behavior of the model matches an observed time series. Probabilistic Programming Process Algebra (ProPPA) \cite{Georgoulas2014} allows some transition rates to be assigned a prior distribution, capturing the modeler’s belief about the likely values of the rates. Using Bayesian inference, prior model is combined with the observations to derive updated probability distributions over parameter values.

Our calibration method falls into the second category: we employ CMA-ES~\cite{Hansen:2005} 
to explore the parameter space and identify the set of parameters that best fit a collection of discrete-time observation trajectories. These trajectories represent empirical averages of specific quantities over time. Beyond the choice of objective function and optimization strategy, a key distinction between our framework and ProPPA is that, in our case, the observations are not generated by a CTMC but rather by a high-fidelity simulator. The CTMC serves as an abstract model of the simulator's behavior. Nevertheless, we empirically demonstrate that our prior model is sufficiently expressive to capture the essential dynamics of the simulator after calibration.


	\section{Conclusion}\label{sec:conclusion}

We have provided the first formal lens on an important industrial
problem.
We have formalized metastability in systems as metastable dynamics in stochastic processes.
Moreover, we have shown how such stochastic models of request-response systems can be constructed through a combination of formal modeling
and data-driven optimization from system descriptions.
The stochastic models provide qualitative (visual) and quantitative predictions about metastable behaviors.
We have shown that computational techniques based on spectral analysis can be used to provide quantitative predictions from the stochastic processes.

%
%
%
We have scratched the surface of metastable dynamics in software systems: 
while we have focused on the important case of request-response systems,
we expect that our definitions, models, and analyses will be applicable to many other instances of metastability in systems.
In a broader context, as mentioned in \cite{analyzing-metastable-faiures-hotos-2025}, the analysis here is one part of a larger
effort to understand and prevent metastable failures in cloud systems.

        \newpage
        \bibliographystyle{ACM-Reference-Format}
	\bibliography{references}

	\newpage
	\appendix
    \section{Detailed Theoretical Results}\label{app:theory_extended}
\subsection{Further Details about CTMCs}
In Section~3.2 of the main paper, we introduced key properties of CTMCs, including the generator matrix, the forward Chapman–Kolmogorov differential equation, and the embedded Markov chain. Here, we provide a more detailed discussion of additional important concepts related to CTMCs. 

In a CTMC $\M$, a state $j\in S$ is \emph{reachable} from state $i\in S$ if $\prob(X(t) = j \mid X(0) = i) = P_{ij}(t) > 0$ for some $t\geq 0$.
States $i$ and $j$ communicate if $i$ is reachable from $j$ and $j$ is reachable from $i$. A CTMC in which every state can be reached from every other state is called an \emph{irreducible} CTMC.
%
Let $\M = (S,Q)$ be an irreducible CTMC with countable state space $S$. $\M$ is called \emph{transient} if for all $x\in S$, 
$\mathbb{P}_x\left(\tau_x<\infty\right)<1$, where $\tau_x$ is the first hitting time of $x$ defined as 
\begin{equation*}
 \label{app:eq:hitting_time}
  \tau_x = \inf \set{ t > 0 \mid X(t) = x}. 
 \end{equation*}
$\M$ is called \emph{recurrent} if it is not transient.
$\M$ is called \emph{positive recurrent} if for all $x \in S$, $\mathbb{E}_x\left[\tau_x\right]<\infty$.
$\M$ is \emph{aperiodic} if its embedded Markov chain is aperiodic.
$\M$ is called \emph{ergodic} if it is irreducible, positive recurrent and aperiodic.
A stationary distribution of the CTMC $\M$ is $\piinf\in[0,1]^{S}$ such that $\piinf Q = 0$. 
A CTMC may have in general more than one stationary distribution. For ergodic CTMCs, the stationary distribution $\piinf$ is unique and that $\lim_{t\rightarrow\infty}\pi(t) = \piinf$ for any initial distribution $\pi_0$.

For an ergodic CTMC $\M$ with stationary distribution $\piinf$, define the distance to the stationary as
\begin{equation*}
	d(t) = \sup_x \|P_{x,\cdot}(t) - \piinf\|_{\textsf{TV}},\quad \forall t\ge 0,
\end{equation*}
where $P_{x,\cdot}(t)$ is the probability distribution over $S$ at time $t$ starting from $x$, and $\|\cdot\|_{\textsf{TV}}$ indicates the total variation distance between two probability distributions defined as 
\begin{equation*}
	\|\pi_1-\pi_2\|_{\textsf{TV}} = \frac{1}{2}\sum_{x\in S} |\pi_1(x)-\pi_2(x)| = \sup_{A\subset S}\|\pi_1(A)-\pi_2(A)\|.
\end{equation*}
For a given $\varepsilon>0$, the \emph{mixing time} $\tmix(\varepsilon)$ of $\M$ is defined as
\begin{equation*}
	\tmix(\varepsilon) = \inf\{t\ge 0\mid d(t)\le \varepsilon\}.
\end{equation*}
\begin{lemma}\label{lem:mixing_time}
	The mixing time can be lower bounded using the following inequality
	\begin{equation}
		\label{eq:mixing_bound}
		\tmix(\varepsilon)\ge \frac{\log(1/2\varepsilon)}{\textsf{Re}(\lambda_{min})}, \quad \forall \varepsilon>0,
	\end{equation}
	where $\lambda_{min}$ is the non-trivial eigenvalue of $Q$ with the smallest real part.
\end{lemma}
\begin{proof}
	This inequality is proved by Montenegro and Tetali~\cite[Theorem 4.9]{MontenegroT05} for CTMCs with $\bar E = \max_i |Q(i,i)| = 1$.
	The same inequality holds for a general $\bar E$. To see this, take any two CTMCs $\M_1 = (S,Q_1)$ and $\M_2 = (S,Q_2)$ with $Q_2 = \alpha Q_1$ for some $\alpha>0$. Using the properties of Chapman-Kolmogorov equation with respect to scaling time, we get that $P_{x,\cdot}^{\M_2}(t) =P_{x,\cdot}^{\M_1}(\alpha t) $ and $d^{\M_2}(t)  = d^{\M_1}(\alpha t)$, which give $\tmix^{\M_2}(\varepsilon) = \frac{1}{\alpha}\tmix^{\M_1}(\varepsilon)$. Then, the left hand side of \eqref{eq:mixing_bound} will be different for $\M_1,\M_2$ by a factor of $1/\alpha$. We also have $\textsf{Re}(\lambda_{min}^{\M_2}) = \alpha\textsf{Re}(\lambda_{min}^{\M_1})$, which gives exactly the same factor to the right hand side of \eqref{eq:mixing_bound}. 
\end{proof}

For a state $x\in S$ and set $D\subset S$, the \emph{escape probability} from $x$ to $D$ is defined as $\prob_x(\tau_D<\tau_x)$.

\subsection{Proof of Theorem~5.2 of the main paper}
	We first show that a CTMC $\M = (S, Q)$ is $\rho$-metastable according to Definition~5.1 in the main paper if and only if its embedded DTMC 
	is $\rho$-metastable.
	To see this, take one realization of the CTMC $X(t)$ and represent it as a sequence of states and their respective holding times $(X_0,H_0,X_0,H_1,\ldots)$. The associated realization of the embedded DTMC is $(X_0,X_1,X_2,\ldots)$. Define the hitting time of any set $A$ in the DTMC with $\bar \tau_A = \inf \set{ n > 0 \mid X_n \in A}$. Then, the hitting time in the CTMC satisfies $\tau_A = \sum_{n=0}^{\bar \tau_A-1} H_n$. This relation gives that for a fixed realization and any two sets $A$ and $B$, $\bar \tau_A<\bar \tau_B$ if and only if $\tau_A<\tau_B$. Then,
	$\mathbb{P}_y\left(\tau_D<\tau_y\right) = \mathbb{P}_y\left(\bar \tau_D<\bar \tau_y\right)$
	and
	$\mathbb{P}_x\left(\tau_{D \backslash x}<\tau_x\right) = \mathbb{P}_x\left(\bar \tau_{D \backslash x}<\bar \tau_x\right)$.
	Therefore, the fraction in Equation~(13) in the main paper will be the same when computed on the CTMC and its embedded DTMC. Similarly, a metastable set $D$ is non-degenerate in the CTMC if and only if it is non-degenerate in the corresponding embedded DTMC. Also note that if $Q$ is multiplied by a constant, both sides of Equation~(14) in the main paper are multiplied by the same constant.
	The rest of the proof follows by applying Theorem 8.43 in \cite{MetastabilityBook} stated for DTMCs to the embedded DTMC of $\M$ and adapting it to the matrix $Q$ of $\M$.
\end{document}